\documentclass[11pt]{article}

\usepackage{amsmath,amssymb}

\usepackage[all]{xy}           
\usepackage{amssymb}           
\usepackage{hyperref}
\usepackage{eucal}
\usepackage{nicefrac}
\usepackage{color}
\usepackage{amscd}
\usepackage{graphicx}
\usepackage[utf8]{inputenc}
\numberwithin{equation}{section}

\newtheorem{definition}{Definition}[section]

\newtheorem{prop}[definition]{Proposition}

\newtheorem{remarkth}[definition]{Remark}

\newenvironment{remark}{\begin{remarkth}\upshape}{\hfill$\diamond$\end{remarkth}}
\def\qed{\ifvmode\removelastskip\fi
{\unskip\nobreak\hfil\penalty50\hbox{}\nobreak\hfil \hbox{\vrule
height1.2ex width1.2ex}\parfillskip=0pt \finalhyphendemerits=0
\par\smallskip}}

\newcommand{\ben}{\begin{enumerate}}
\newcommand{\een}{\end{enumerate}}

\newcommand{\ds}{\displaystyle}

\def\derpar#1#2{\ds\frac{\partial{#1}}{\partial{#2}}}

\usepackage{color}

\def\bcdb{\begin{color}{darkblue}}
\def\bcr{\begin{color}{red}}
\def\bcb{\begin{color}{blue}}
\def\enc{\end{color}}
\def\r{\ensuremath{\mathbb{R}}}
\def\rk{{\mathbb R}^{k}}
\def\tkq{T^1_kQ}
\def\lto{\longrightarrow}

\def\tm{T^1_kQ\oplus(T^1_k)^*Q}
\def\tkqh{(T^1_k)^*Q}

\def\to{\rightarrow}
\def\lto{\longrightarrow}

\def\r{\mathbb{R}}
\def\rk{\mathbb{R}^k}
\def\r{\mathbb{R}}
\def\tkq{T^1_kQ}

\def\r{\ensuremath{\mathbb{R}}}
\def\rk{{\mathbb R}^{k}}

\def\tabaddress#1{{\small\it\begin{tabular}[t]{c}#1
\\[1.2ex]\end{tabular}}}
\def\proof{ ({\sl Proof\/}) }

\title{Tulczyjew's derivations and intrinsic field equations in classical field theories}
\author{\sc Modesto Salgado\thanks{{\bf e}-{\it mail}:
  modesto.salgado@usc.es}
  \\
  \tabaddress{Departamento de Matem\'aticas  \\
 Facultade de Matem\'{a}ticas,
    Universidade de Santiago de Compostela,\\
    15706-Santiago de Compostela, Spain}
   \\
 {\sc      Silvia Vilari\~no\thanks{{\bf
e}-{\it mail}:
silviavf@unizar.es}} \\
  \tabaddress{Centro Universitario de la Defensa de Zaragoza $\&$ I.U.M.A.  \\
 Academia General Militar,
  Carretera de Huesca s/n\\
    50090 Zaragoza, Spain}}

\begin{document}

\maketitle

\pagestyle{myheadings}
\parskip=7pt

\thispagestyle{empty}


\begin{abstract} 
  
  This work presents the variational principles and the intrinsic versions of several equations in field theories, in particular, for the Classical Euler-Lagrange field equations, the implicit Euler-Lagrange field equations and the non-holonomic implicit Euler-Lagrange field equations. The advantages of the variational and intrinsic versions of these equations is that the Lagrangians functions are not necessary regular Lagrangians. We present two examples of this situation: Navier's equations and the non-holonomic Cosserat rod. Finally we comment the Hamiltonian case when the Lagrangian is a hyperregular function.

\end{abstract}

 \maketitle

\tableofcontents

 \newpage
\mbox{}
\thispagestyle{empty} 

 \section{Introduction}\label{sec:intro}
 
   Classical field theories  are physical theories that describe the behaviour of one or more physical fields through field equations. Recall that a physical field can be thought of as the assignment of a physical quantity at each point of space and time.
    
At present, it is common to model classical field theories using different mathematical formalisms. For instance, in the literature there exist many geometric models that describe classical  field theories. Just to name a few of them: the polysymplectic \cite{Kana,Sarda_1993}, the $n$-symplectic \cite{Norris_2001},  the $k$-cosymplectic \cite{LMeS-2001},  the
multisymplectic  \cite{CCI-91,EMR-1996,Gymmsy,KT-79} and the jet  formalisms \cite{Saunders-89}.

The main differences between all these models depend on the choice one makes for the geometric and the differentiable structure of both the space of parameters $x^\alpha$ (such as space-time) and the space of fields $\phi^i$. The model we will use in this paper is the one of $k$-symplectic field theory, as developed in  the papers \cite{Awane-1992,LSV-2015,MRS-2004,Gu-1987,RSV-2007}.

 Let us recall that the $k$-symplectic   formalism \cite{LSV-2015}    is the generalization to
first order classical field theories of the standard symplectic formalism in
mechanics, which is the geometric framework for describing autonomous dynamical
systems.
   
   The k-symplectic formalism is used to give a geometric
description of certain kinds of field theories: in a local description, those whose Lagrangians do not depend on the coordinates in the basis (in many of them, the space-time coordinates); that is, it is only valid for Lagrangians $L(q^i,v^i_\alpha)$   and Hamiltonians $H(q^i,p^\alpha_i)$  that depend on the field coordinates $q^i$ and on the partial derivatives of the field $v^i_\alpha$ or the corresponding momenta $p^\alpha_i$. Thus we consider the   Lagrangian and Hamiltonian functions as   maps $L:\tkq \to \r$, and $H:\tkqh \to \r$. This formalism characterizes the (regular) field theory in terms of a certain class of so-called `$k$-vector fields' on $T^1_kQ$ and $(T^1_k)^*Q$, which are literally collections of $k$ individual vector fields.

Although we do not strictly use this formalism in this work, if we use the bundles $\tkq$ and $\tkqh$, which are part of this formalism. The idea is to give  new descriptions of several equations of fields theories, using geometric structures but not the $k$-vector fields. The advantages of these new descriptions are that this description is valid when the Lagrangian is singular or when we consider non-holonomic constraints.

 Explicitly, the main  aim of this paper is to present the variational principles and the  intrinsic versions of the 
\begin{quote} 
 \begin{enumerate}
 
 \item  Euler-Lagrange field equations, 
 
  \item  Implicit Euler-Lagrange field equations,
  
  \item  Non-holonomic implicit Euler-Lagrange field equations.
  
  \item  Hamilton-de Donder-Weyl equations and  non-holonomic
 Hamilton-de Donder-Weyl equations.
  
   \end{enumerate}
  \end{quote} 
   
   The organization of the paper is as follows:  
  Section \ref{sec:geo_elem}  introduces the necessary bundles along this work, that is:
   \begin{itemize}
  \item   The   tangent bundle of $k^1$-velocities        $T^1_kQ=TQ\oplus_Q\stackrel{k}{\dots}\oplus_Q TQ$, that is  the Whitney sum of $k$ copies of the tangent bundle $TQ$ of a manifold $Q$.  
   \item  The cotangent bundle of $k^1$-velocities $(T^1_k)^*Q= T^*Q\oplus_Q\stackrel{k}{\dots}\oplus_Q T^*Q$, that is the Whitney sum of $k$ copies of the cotangent bundle $T^*Q$.
    \item  The generalized Pontryagin bundle $\mathcal{M}= T^1_kQ\oplus_{Q} (T^1_k)^*Q$, which it is necessary  in order to obtain (non-holonomic) implicit Euler-Lagrange equations.
     \end{itemize}

   In the tangent and cotangent  bundles of $k^1$-velocities and   $k^1$-covelocities,  we have canonical  forms, and we can define natural  prolongations of maps and vector fields which will are fundamental  to develop  the  main aim in  section \ref{sec:EL_eq}. 
   
   In Section \ref{sec:Tulczyjew}, defining an  extension of     the   Tulczyjew's derivations  \cite{T1,T2,T3},   we obtain two $1$-forms  $\lambda$ and $\chi$, on $T^1_k(\tkqh)$, which are also fundamental  for giving the intrinsic version of the field equations in  section \ref{sec:EL_eq}.
  
    With all these tools in Section \ref{sec:EL_eq} and Section \ref{sec:HDWeq} we establish the   variational principles and the  intrinsic versions of the corresponding field equations. For the   intrinsic version of  (i)  we consider a Lagrangian $L$ and the $1$-form   $\lambda$;  for (ii)  and (iii) we use the generalized energy function $E\colon \mathcal{M}\to \mathbb{R}$ and the $1$-form   $\chi$.  In Section \ref{sec:HDWeq}, we describe  an intrinsic version of  the equations (iv), using   $E$ and $\chi$.
    
    As example of point (ii) we describe the Navier's equations, and as example of point (iii) we describe the nonholonomic Cosserat rod. 
    
    Some results of Sections \ref{sec:EL_eq} and \ref{sec:HDWeq} , considering the case $k=1$,  are 
    generalizations of some results given \cite{YM-2006a,YM-2006b}.

 Section \ref{sec:conclusions} summarises the results of the work and provides some hints on future research. 
 
 Finally, unless otherwise stated, we assume all mathematical objects to be real, smooth and globally defined.

\section{Geometric elements}\label{sec:geo_elem}
In this section, we provide a quick overview of the natural bundles for the study of field theories using the geometric elements of the $k$-symplectic setting \cite{Awane-1992,LSV-2015,MRS-2004,Gu-1987,RSV-2007}. More details about the material of this section can be found in \cite{LSV-2015} and the references therein. Using these two bundles we define the ``$k$-Pontryagin bundle'' $\mathcal{M}\colon =T^1_kQ\oplus (T^1_k)^*Q$, one generalization of the usual ``Pontryagin bundle'' $TQ\oplus_Q T^*Q$. This manifold $\mathcal{M}$ is very important along this paper in particular in the description of the implicit version of the Euler-Lagrange field equations (with or without non-holonomic constraints).

\subsection{The tangent bundle of $k^1$-velocities}\label{subsec:k-tangent}

Let $\tau_M\colon TM \to M$ be the tangent bundle of a differentiable manifold $M$. We will use the notation $T^1_kM$ for the Whitney sum
$TM\oplus\stackrel{k}{\dots}\oplus TM$ of $k$ copies of $TM$ and $\tau^k_M$
 for the corresponding projection $\tau^k_M \colon T^1_kM\to M$ which maps $(v_{1_m},\ldots , v_{k_m})$ onto the point $m\in M$.
 
$T^1_kM$ can be identified with the manifold $J^1_0(\mathbb{R}^k,M)$ of $k^1$-velocities of
$M$. These are $1$-jets of maps from $\rk$ to $M$ with source at $0\in \mathbb{R}^k$.
For this reason the manifold $T^1_kM$ is called {\sl the tangent bundle of $k^1$-velocities of $M$}.

In what follows, we will denote coordinates on $\r^k$ by  $ (x^\alpha) = (x^1,\ldots, x^k)$. If $(y^I)$ (with $I=1,\ldots , \dim M$) are local coordinates on $U \subset M$ then the
induced local coordinates   $(y^I,
u^I)$ on $TU=\tau_M^{-1}(U)$ are given by
$$
y^I( v_m)=y^I(m),\qquad u^I(v_m)=v_m(y^I)\, , \quad v_m\in T_mM.
$$

These naturally induce coordinates $(y^I ,u^I_\alpha)$ (with $I=1,\ldots ,  \dim M;\, \alpha=1,\ldots, k$) for a point $(v_{1_m},\ldots , v_{k_m})$ in 
$T^1_kU=(\tau_M^k)^{-1}(U)$,
such that  $u^I_\alpha$ are the components of the $\alpha$'th vector $v_{\alpha_m}$  along the natural basis of $T_mM$
$$
v_{\alpha_m}= u^I_\alpha \, \derpar{}{y^I}\Big\vert_m,
$$ that is,
\begin{equation}\label{coortkq}
y^l(v_{1_m},\ldots , v_{k_m})=y^l(m)\, , \quad   u^I_\alpha(v_{1_m},\ldots , v_{k_m})=v_{\alpha_m}(y^I).
\end{equation}

The canonical projection $\tau^{k}_M\colon T^1_kM \rightarrow M$  
is given in local coordinates as follows
\begin{equation}\label{tk00}
\tau^{k}_M(y^I ,u^I_\alpha)=(y^I) \, .
\end{equation}

On the other hand,  we have a family of  canonical projections $\tau_M^{k,\alpha}:T^1_kM\to TM$ defined for each $\alpha\in \{1,\ldots, k\}$ by
\begin{equation}\label{map:taukalphaalpha}
\tau_M^{k,\alpha} (v_{1_m},\ldots , v_{k_m})
=
    {v_\alpha}_m\,,
\end{equation}
given  in local coordinates by
\begin{equation}\label{loctkalfa}
\tau_M^{k,\alpha}(y^I ,u_1^I,\ldots, u_k^I)=(y^I,u_\alpha^I ).
\end{equation}

We continue this subsection recalling some geometric elements defined on the tangent bundle of $k^1$-velocities, which will be important along this work. These elements are the canonical prolongations of maps, the complete lifts of vector fields and finally the first prolongation of maps.

\paragraph{A. Canonical prolongations of maps  $\varphi\colon M\to N$}\

Let $\varphi\colon M\to N$ be a differentiable map. In what follows, we will  make use of the {\sl canonical prolongation of $\varphi$}, which is  the induced
map $T^1_k\varphi:T^1_kM \to  T^1_k N$  defined by
\[
  T^1_k\varphi(v_{1_m},\ldots , v_{k_m})=
(\varphi_*(m)(v_{1_m}),\ldots,\varphi_*(m)(v_{k_m})) \,  .
\]

 \paragraph{B. Complete lifts of vector fields}\
 
We now recall the notion of  canonical prolongation of a vector field $X\in \mathfrak{X}(M)$ to $T^1_kM$, that is the   complete lift $X^C\in\mathfrak{X}(T^1_kM)$. If $X$ has a local $1$-parametric group
of transformations $\varphi_t \colon Q \to Q$, then the local
$1$-parametric group of transformations $T^1_k\varphi_t\colon T^1_kM
\to T^1_kM$ generates a vector field $Z^C$ on $T^1_kM$, the {\sl complete lift} of $X$ to $ T^1_kQ$. Its local expression is
\begin{equation} \label{clift}
X^C=X^I\frac{\partial}{\partial y^I}+u^I_\alpha \frac{\partial X^J}
{\partial y^I}\frac{\partial}{\partial u^J_\alpha} \,.
\end{equation}

 \paragraph{C. First prolongation of maps $\psi: \r^k \to M$}\

The {\sl first prolongation $\psi^{(1)}$ of a map $\psi: \r^k \to M$} is the map $\psi^{(1)}\colon\r^k\to T^1_kM$, defined by
\[
\psi^{(1)}(x)=
 \left(\psi_*(x)\left(\derpar{}{x^1}\Big\vert_x\right),\ldots,
\psi_*(x)\left(\derpar{}{x^k}\Big\vert_x\right)\right) \, .
\]
In local coordinates, we have
\begin{equation}\label{localfi11}
\psi^{(1)}(x)=\left( \psi^I (x), \frac{\partial\psi^I}{\partial x^\alpha} (x)\right), \qquad  1\leq \alpha\leq k\, ,\, 1\leq I\leq \dim M \, .
\end{equation}

\subsection{The cotangent bundle of $k^1$-covelocities}\label{subsec:k-cotangent}

Let $\pi_M\colon T^*M \to M$ be the cotangent bundle of the manifold $M$. We will use the notation $(T^1_k)^*M$ for the Whitney sum
$T^*M\oplus\stackrel{k}{\dots}\oplus T^*M$ of $k$ copies of $T^*M$ and $\pi^k_M$
 for the corresponding projection $\pi^k_M \colon (T^1_k)^*M\to M$ which maps $(\nu^1_m,\ldots ,
\nu^k_m)$ onto the point $m\in M$.

$(T^1_k)^*M$ can be identified with the manifold $J^1(M,\mathbb{R}^k)_0$ of $k^1$-covelocities of
$M$. These are $1$-jets of maps from $M$ to $\rk$ with target at $0\in \mathbb{R}^k$.
For this reason the manifold $(T^1_k)^*M$ is called {\sl the cotangent bundle of
$k^1$-covelocities of $M$}.

If $(y^I)$ are local coordinates on $U \subset M$ then the
induced local coordinates     $(y^I,
z_I)$ on $T^*U=\pi_M^{-1}(U)$ are given by
$$
y^I( \nu_m)=y^I(m)\, ,\qquad z_I(\nu_m)=\nu_m(\derpar{}{y^I}\Big\vert_m), \quad \nu_m\in T^*_mM.
$$

These naturally induce coordinates $(y^I ,z_I^\alpha)$  for a point $(\nu^1_m,\ldots ,
\nu^k_m)$ in 
$(T^1_k)^*U=(\pi_M^k)^{-1}(U)$,
such that  $z^\alpha_I$ are the components of the $\alpha$'th covector ${\nu^\alpha_m}$  along the natural basis of $T_mM$
  that is
\begin{equation}\label{coortkqh}
y^l(\nu^1_m,\ldots ,
\nu^k_m)=y^l(m), \quad   z_I^\alpha(\nu^1_m,\ldots ,
\nu^k_m)=
\nu^\alpha_m(\derpar{}{y^I}\Big\vert_m)\,.
\end{equation}

The canonical projection $\pi^{k}_M\colon (T^1_k)^*M \rightarrow M$  
is given in local coordinates as follows
\begin{equation}\label{tkalfa}
\pi^k_M(y^I ,z_I^\alpha)=(y^I )\, ,
\end{equation}
and  
for each $\alpha\in \{1,\ldots, n\}$,   the  canonical projections $\pi_M^{k,\alpha}:(T^1_k)^*M\to T^*M$ are defined by
\begin{equation}\label{map:taukalpha}
\pi_M^{k,\alpha} (\nu_{1_m},\ldots , \nu_{k_m})
=
    {\nu_\alpha}_m \, ,
\end{equation}
and its local expression is 
\begin{equation}\label{loctkqalfa}
\pi_M^{k,\alpha}(y^I ,z_I^1,\ldots, z^k_I)=(y^I,z^\alpha_I )\, .
\end{equation}

As in the case of the tangent bundle of $k^1$-velocities, we now recall the canonical prolongations of maps and vector field.
  
  \paragraph{A. Canonical prolongations of maps  $\varphi\colon M\to N$}\

  Let $\varphi\colon M \to N$	a map. The natural or canonical prolongation of $(T^1_k)^*\varphi$ to the corresponding bundles of $k^1$-covelocities is the map 
$
(T^1_k)^*\varphi\colon (T^1_k)^*N \to (T^1_k)^*M
 $
 defined as follows:
 \begin{equation}\label{exp:k-cotangent_prolongation00}
  (T^1_k)^*\varphi(\nu_{\varphi(m)}))=(\varphi^*(\nu_{1_{\varphi(m)}}), \ldots, \varphi^*(\nu_{k_{\varphi(m)}}))= (\nu_{1_{\varphi(m)}}\circ \varphi_*(m), \ldots,\nu_{k_{\varphi(m)}}\circ \varphi_*(m))\,,
 \end{equation}
 where $\nu_{\varphi(m)}=(\nu_{1_{\varphi(m)}}, \ldots, \nu_{k_{\varphi(m)}})\in (T^1_k)^*N$ and $m\in M$.

 \paragraph{B. Complete lifts of vector fields}\
 
Now let $X$ be a vector  field on $M$  with  local $1$-parametric group
of transformations $\varphi_t \colon M \to M$, then the local
$1$-parametric group of transformations $(T^1_k)^*\varphi_t\colon (T^1_k)^*M
\to (T^1_k)^*M$ generates a vector field $X^{C^*}$ on $T^1_kM$, the {\sl complete lift} of $X$ to $ (T^1_k)^*M$. Its local expression is
\begin{equation} \label{clifth}
X^{C^*}=X^I\frac{\partial}{\partial y^I}-z^\alpha_J \frac{\partial X^J}
{\partial y^I}\frac{\partial}{\partial z^\alpha_I} \, .
\end{equation}

\subsection{The Pontryagin bundle}\label{subsec:Pontryagin}

In order to introduce the Hamilton-Pontryagin principle  we define the  generalized Pontryagin bundle $\mathcal{M}= T^1_kQ\oplus_{Q} (T^1_k)^*Q$. This bundle plays a similar role as the Pontryagin bundle $TQ\oplus_QT^*Q$ over a configuration manifold $Q$ for the case of classical mechanics.

In this section we consider the geometric elements over this bundle, which are necessary in the rest of the paper.

Let us consider the Whitney sum $\mathcal{M}= T^1_kQ\oplus_Q (T^1_k)^*Q$ of the tangent bundle of $k^1$-velocities and the cotangent bundle of  $k^1$-covelocities of a differentiable manifold $Q$. This manifold is called the \textit{$k$-Pontryagin bundle}.

An element of $\mathcal{M}$ is a pair $(\rm{v}_q,\nu_q)$ where $\rm{v}_q=({v_1}_q,\ldots, {v_k}_q)\in T^1_kQ$ and $\nu_q=(\nu^1_q, \ldots, \nu^k_q)\in (T^1_k)^*Q$. It has natural bundle structures over $T^1_kQ$ and $(T^1_k)^*Q$. 

Let us denote by $pr_1\colon \mathcal{M}=T^1_kQ\oplus_Q (T^1_k)^*Q\to T^1_kQ$ the projection into the first factor, $pr_1(\rm{v}_q,\nu_q)=\rm{v}_q$ and by $pr_2\colon \mathcal{M}=T^1_kQ\oplus_Q (T^1_k)^*Q\to (T^1_k)^*Q$ the projection into de second factor, $pr_2(\rm{v}_q,\nu_q)=\nu_q$. 

We denote by $pr^\mathcal{M}_Q\colon \mathcal{M}=T^1_kQ\oplus_Q (T^1_k)^*Q\to Q$ the projection into the configuration space $Q$, $pr^\mathcal{M}_Q(\rm{v}_q,\nu_q)=q$.

  Taking into account (\ref{coortkq})  and (\ref{coortkqh}) each coordinate system $(y^I)\equiv(q^i)$ defined on an open neighbourhood $U\subset Q$, induces the local bundle coordinate system 
$(y^I,z^I_\alpha)\equiv(q^i,v^i_\alpha)$ on $(\tau^k_Q)^{-1}(U)$, the local bundle coordinate system  $(y^I,z^\alpha_I)\equiv(q^i, p^\alpha_i)$ on $(\pi^k_Q)^{-1}(U)$ and  
$(q^i,v^i_\alpha, p^\alpha_i)$  on  $(pr^{\mathcal{M}}_Q)^{-1}(U)$ defined as follows:
\begin{equation}\label{eq:coord_M}
q^i(\rm{v}_q,\nu_q)=q^i(q)\, ,\quad v^i_\alpha(\rm{v}_q,\nu_q)=v^i_\alpha(\rm{v}_q)= v_{\alpha_q}(q^i)\, ,\quad p^\alpha_i(\rm{v}_q,\nu_q)=p^\alpha_i(\nu_q)=\nu^\alpha_q\Big(\displaystyle\frac{\partial}{\partial q^i}\Big\vert_{q}\Big)\,.
\end{equation}

These coordinates endow to $\mathcal{M}$ of a structure of differentiable manifold of dimension $n(2k+1)$.

\subsubsection{Canonical prolongations of diffeomorphisms and vector fields }

Using the definition of the tangent and cotangent map we introduce the prolongation of a diffeomorphism.

Let $\varphi\colon Q\to Q$ be a diffeomorphism. The natural or canonical prolongation of $\varphi$ to the corresponding $k$-Pontryagin bundles is the map
\[
 \tau^1_k\varphi\colon \mathcal{M}=T^1_kQ\oplus_Q (T^1_k)^*Q\to \mathcal{M}=T^1_kQ\oplus_Q (T^1_k)^*Q
\] defined by
\begin{equation}\label{exp:prontryagin_prolongation}
\tau^1_k\varphi(\rm{v}_q,\nu_q)=(T^1_k\varphi(\rm{v}_q), (T^1_k)^*\varphi(\nu_q))\,,
\end{equation}
where $T^1_k\varphi$ and $(T^1_k)^*\varphi$ are the natural prolongations of $\varphi$ introduced in  Sections \ref{subsec:k-tangent} and \ref{subsec:k-cotangent}, respectively.

The above definition allows us to introduce the canonical or complete lift of vector fields from $Q$ to $\mathcal{M}=T^1_kQ\oplus_Q(T^1_k)^*Q$.
\begin{definition}\label{def:complete_pontryagin_lift}
Let $Z$ be a vector field on $Q$, with $1$-parameter group of diffeomorphism $\{\varphi_s\}$. The canonical o complete lift of $Z$ to the $k$-Pontryagin bundle $\mathcal{M}=T^1_kQ\oplus_Q(T^1_k)^*Q$ is the vector field $Z^1$ whose local $1$-parameter group of diffeomorphism is $\{\tau^1_k\varphi_s\}$.
\end{definition}

In local canonical coordinates (\ref{eq:coord_M}), if $Z=Z^i\nicefrac{\partial}{\partial q^i}$, the local expression of $Z^1$ is
\begin{equation}\label{local:prontryagin_lift}
 Z^1=Z^i\displaystyle\frac{\partial}{\partial q^i} + v^j_\alpha\displaystyle\frac{\partial Z^k}{\partial q^j}\displaystyle\frac{\partial}{\partial v^k_\alpha} - p^\alpha_j\displaystyle\frac{\partial Z^j}{\partial q^k}\displaystyle\frac{\partial}{\partial p^\alpha_k}\,.
\end{equation}
Compare this local expression with  (\ref{clift}) and (\ref{clifth}).
\subsubsection{Canonical forms on   $\mathcal{M}$ }
We now introduce certain canonical forms on $(T^1_k)^*Q$ and  $\mathcal{M}$. We consider the  \emph{canonical $1$-forms} $\Theta^1,\ldots, \Theta^k$ on $(T^1_k)^*Q$ as the pull-back of Liouville's $1$-form $\theta$ by the canonical projection $\pi_Q^{k,\alpha}:\tkqh \to T^*Q$ , that is, for each $1\leq \alpha\leq k $
\begin{equation}\label{eq:Theta_alpha}
\Theta^\alpha = (\pi^{k,\alpha})^*\theta\,;
\end{equation}
the \emph{canonical $2$-forms} $\Omega^1,\ldots, \Omega^k$ are defined by
\begin{equation}\label{eq:Omega_alpha}
\Omega^\alpha=-d\Theta^\alpha \, ,
\end{equation}
or equivalently by $\Omega^\alpha=(\pi^{k,\alpha})^*\omega$ being $\omega$ the canonical symplectic form on the cotangent bundle $T^*Q$.

If we consider the canonical coordinates $(q^i,p^\alpha_i)$ on  $(T^1_k)^*Q$  then  the canonical forms $\Theta^\alpha, \Omega^\alpha$ have the following local expressions:
\begin{equation}\label{eq:Locformcan}
\Theta^\alpha=p^\alpha_idq^i\,,\quad \Omega^\alpha=dq^i\wedge dp^\alpha_i\,.
\end{equation}
  
We will also consider    the forms on $\mathcal{M}$
\begin{equation}\label{eq:can_form_M}
\Theta^\alpha_{\mathcal{M}}=(pr_2)^*(\Theta^\alpha)\, ,\quad \Omega^\alpha_{\mathcal{M}}=(pr_2)^*(\Omega^\alpha)\, ,\quad 1\leq \alpha\leq k\,,
\end{equation}
with local expressions
\begin{equation}\label{eq:local_can_form_M}
\Theta^\alpha_{\mathcal{M}}=p^\alpha_idq^i\, ,\quad \Omega^\alpha_{\mathcal{M}}=dq^i\wedge dp^\alpha_i\, ,\quad 1\leq \alpha\leq k\,.
\end{equation}
\subsubsection{Generalized energy function}

Let $L\colon T^1_kQ\to \mathbb{R}$ be a Lagrangian function, which is possible degenerate. We define the \textit{generalized energy function} associated to $L$ by the map $E\colon \mathcal{M}=T^1_kQ\oplus_Q (T^1_k)^*Q\to \mathbb{R}$ defined as
\begin{equation}\label{eq:energy}
 E(\rm{v}_q,\nu_q) = <<\rm{v}_q,\nu_q,>> - \,L(\rm{v}_q)\,,
\end{equation}
for each $(\rm{v}_q,\nu_q)\in \mathcal{M}=T^1_kQ\oplus_Q (T^1_k)^*Q$. In the above definition $<<\cdot, \cdot>>\colon T^1_kQ\oplus (T^1_k)^*Q\to \mathbb{R}$ is the map defined by
\begin{equation}\label{def:pairing}
<<\rm{v}_q,\nu_q,>>=\nu_q(\rm{v}_q)=\displaystyle\sum_{\alpha=1}^n \nu^\alpha_q(v_{\alpha_q})\,.
\end{equation}

In the induced local coordinates system (\ref{eq:coord_M}) we obtain
\begin{equation}\label{local:energy}
E(q^i, v^i_\alpha, p^\alpha_i)= p^\alpha_i v^i_\alpha- L(q^i, v^i_\alpha)\,.
\end{equation}

\subsection{The Legendre transformation}\label{subsec:legendre}

In order to stablish a relationship between the Lagrangian and Hamiltonian version of the equations of classical field theories, we recall the definition of the Legendre transformation $FL$ between the tangent bundle of $k^1$-velocities and the cotangent bundle of $k^1$-covelocities.

For each Lagrangian function $L\in \mathcal{C}^\infty(T^1_kQ)$ it is possible to consider  the \textit{Legendre transformation} associated to $L$ as the map  $FL:T^1_kQ\to (T^1_k)^*Q$ defined as follows:
\[
FL({\rm v}_q)= ( [FL({\rm v}_q)]^1,\ldots ,[FL({\rm v}_q)]^k )
\]
where
\[ [FL({\rm v}_q)]^\alpha(u_{q})=
\displaystyle\frac{d}{ds}\Big\vert_{s=0}\displaystyle L
\left(
{v_1}_{q}, \dots,{v_\alpha}_{q}+su_{q}, \ldots, {v_k}_{q}
\right)\,,
\]
for  $1\leq \alpha\leq k $, $u_{q}\in T_{q}Q$ and ${\rm v}_q=({v_1}_{q},\ldots, {v_k}_{q})\in T^1_kQ$.

Using natural coordinates $(q^i, v^i_\alpha)$ on $T^1_kQ$ and $(q^i, p^\alpha_i)$ on $(T^1_k)^*Q$, the local expression of the Legendre map is
\begin{equation}\label{map:LocLegtkq}
FL(q^i, v^i_\alpha) =\Big(q^i, \frac{\displaystyle\partial L}{\displaystyle\partial v^i_\alpha } \Big)\, .
\end{equation}

Let us recall that a a Lagrangian function  $L: T^1_kQ\longrightarrow \mathbb{R} $ is said to be  \emph{regular} (resp. \emph{hyperregular}) if the Legendre map $FL$ is a local diffeomorphism (resp. global). In other case  $L$ is said to be  \emph{singular}.

From (\ref{map:LocLegtkq}) we know that $L$ is regular if and only if the matrix $\Big(\displaystyle\frac{\partial^2 L}{\partial v^i_\alpha\partial v^j_\beta}\Big)$ is not singular.
\section{Tulczyjew's derivations  and canonical forms}\label{sec:Tulczyjew}
One of the aim of this paper is to obtain an alternative description of the Lagrangian and Hamiltonian field equations. In a similar description on the case of Lagrangian and Hamiltonian Mechanics it is possible to obtain a symplectic structure on $TT^*Q$. This symplectic form can be defined by two different ways as the exterior derivative of two intrinsic one-forms on $TT^*Q$.

 The aim of this section is to extend that construction and obtain  two intrinsic $1$-forms $\chi$ and $\lambda$  on the space $T^1_k((T^1_k)^*Q)$.

In order to define these two $1$-forms it is necessary to consider Tulczyjew's derivations.

\subsection{Tulczyjew derivations on $T^1_kM $}

 Let us denote by $\bigwedge N$   the algebra
of the exterior differential forms on an arbitrary manifold $N$. In
\cite{T1,T2}, a derivation $$i_{T}: \bigwedge M\to \bigwedge TM $$ of degree $-1$   over   the canonical projection  $\tau_M:TM\to M$ was defined in
an arbitrary manifold $M$ by $i_{T}\mu=0$ if $\mu$ is a function on
$M$, and by
$$
i_{T}\mu(v_x)  (Z^1_{v_x},\ldots ,Z^l_{v_x})=\mu
(x)(v_x,(\tau_M)_*(v_x)(Z^1_{v_x}),\ldots
,(\tau_M)_*(v_x)(Z^l_{v_x}))\, ,$$
if $\mu$ is a $(l+1)$-form, where $x\in M$, $Z^r_{v_x}\in T_{v_x}(TM), \, 1\leq r\leq l$.

A derivation $$d_T : \bigwedge M\to \bigwedge TM $$   of degree $0$   over $\tau_M$ is defined by $d_T\mu=i_Td\mu +di_T\mu$, where $d$
is the exterior derivative. We have $dd_T=d_Td$.

  We extend the above definitions   of $i_T$ and $d_T$  as follows: for every $\alpha=1 ,
\ldots , k$ we  define a derivation $$i_{\displaystyle T_\alpha}  : \bigwedge M\to \bigwedge T^1_kM $$   of degree $-1$
  over $\tau^k_M:T^1_kM\to M$
by $i_{T_\alpha}\mu=0$ if $\mu$ is a function on $M$, and by
\begin{equation}\label{itmua}
i_{\displaystyle T_\alpha}\mu(w_x)  (\widetilde{Z}^1_{w_x},\ldots
,\widetilde{Z}^l_{w_x})=\mu
(x)(  \tau_M^{k,\alpha} (w_x),(\tau^k_M)_*(w_x)(\widetilde{Z}^1_{w_x}),\ldots
,(\tau^k_M)_*(w_x)(\widetilde{Z}^l_{w_x}))\, ,\end{equation}
if $\mu$ is an $(l+1)$-form,  $w_x\in T^1_kM$
and $\widetilde{Z}^r_{w_x}\in T_{w_x}(T^1_kM), \quad 1\leq r \leq
l $.

  We define, for each  $\alpha=1 ,
\ldots , k$,  a  derivation $$d_{T_\alpha}: \bigwedge M \to \bigwedge T^1_kM$$ of degree $0$   over $\tau^k_M$ is defined by $$d_{T_\alpha}\mu=i_{T_\alpha}
d\mu +di_{T_\alpha}\mu\, , $$  where $d$ is the exterior derivative. We have
$d\, d_{T_\alpha}=d_{T_\alpha}d$.

\subsection {Canonical $1$-forms $\chi, \lambda$ on $T^1_k((T^1_k)^*Q)$}\label{chilambda}

We now consider the above definitions with $M=(T^1_k)^*Q$.  Then the Tulczjew's derivations on $T^1_kM=T^1_k((T^1_k)^*Q)$ are the following maps:
$$d_{T_\alpha}: \bigwedge (T^1_k)^*Q \to \bigwedge T^1_k((T^1_k)^*Q)
\qquad  i_{T_\alpha}: \bigwedge (T^1_k)^*Q \to \bigwedge T^1_k((T^1_k)^*Q)$$

   With  the canonical $1$-forms $\Theta^\alpha$ on $(T^1_k)^*Q$   we can define the intrinsic $1$-form  on $T^1_k((T^1_k)^*Q)$
\begin{equation}\label{eq:lambda}
 \lambda=\displaystyle\sum_{\alpha=1}^kd_{T_\alpha}\Theta^\alpha\,.
\end{equation}

In a similar way we can define another intrinsic $1$-forms    on $T^1_k((T^1_k)^*Q)$ using the derivations of degree $-1$ and the family of canonical $2$-forms $\Omega^1,\ldots, \Omega^k$ 
 \begin{equation}\label{eq:chi}
 \chi = \displaystyle\sum_{\alpha=1}^k\iota_{T_\alpha}\Omega^\alpha\,.
\end{equation}

Since on $\tkqh$ we have local coordinates $(y^I)\equiv(q^i,p^\alpha_i)$,      we have the induced coordinates $ (y^I,z_\alpha^I)\equiv(q^i,p^\alpha_i,(v_\alpha)^i,(v_\alpha)^\beta_i)
$ on $T^1_k(\tkqh)$.

Using a computation in these local coordinates we obtain that the local expressions of $\lambda$ and $\chi$ are
\begin{equation}\label{eq_local_lambda}
\lambda=  \displaystyle\sum_{\alpha=1}^kd_{T_\alpha}\Theta^\alpha=(v_\alpha)^\alpha_i dq^i+ p^\alpha_i\, d(v_\alpha)^i \,,
\end{equation}
and
\begin{equation}\label{eq:local_chi}
\chi =   \displaystyle\sum_{\alpha=1}^k\iota_{T_\alpha}\Omega^\alpha =(v_\alpha)^i\, \,  dp^\alpha_i - (v_\alpha)^\alpha_i dq^i.
\end{equation}

\section{The intrinsic form of the Euler-Lagrange field equations}\label{sec:EL_eq}

In this section we describe the implicit Euler-Lagrange field equations in two different ways. In first place we obtain the implicit Euler-Lagrange equations for classical field theories from a variational principle. Then, we describe the intrinsic form of these equations using the canonical forms $\chi$ and $\lambda$ defined in Section \ref{sec:Tulczyjew}. 

The sketch of this Section is to give the variational principle and the intrinsic form of the Euler-Lagrange field equations in three different cases: we recall the classical case, we describe the implicit Euler-Lagrange field equations without constraints and, finally, the non-holonomic implicit Euler-Lagrange field equations.
\subsection{Classical Euler-Lagrange field equations}

In this subset we recall the classical Euler-Lagrange equations in field theories. In particular we recall the variational description of the  Euler-Lagrange field equations in the $k$-symplectic setting \cite{LSV-2015}  and then, we consider the intrinsic version of these equations.
\subsubsection{The Hamilton principle}


 Consider a Lagrangian function  $L:T^1_kQ  \to \r$.  We now define the  {\it action integral}
 \[{\mathcal J}(\phi)=\ds\int_{U_0} (L\circ \phi^{(1)})(x) d^kx\,,\]
 where $d^kx=dx^1\wedge\ldots\wedge dx^k$ is a volume form on $\rk$, $\phi:U_0\subset\rk\to Q$ is a map, with compact support,
 defined on an open set $U_0$ and $\phi^{(1)}:U_0\subset\rk\to T^1_kQ$
 denotes the first prolongation of $\phi$, introduced in Section \ref{subsec:k-tangent}.
 
  A map $\phi$ is called an
 extremal for the above action if
 \[\ds\frac{d}{ds}\mathcal{J}(\tau_s\circ\phi)\Big\vert_{s=0}=0
\]for every flow $\tau_s$ on $Q$ such that $\tau_s(q)=q$ for all
 $q$ in the boundary of $\phi(U_0)$. Since such a flow $\tau_s$ is
 generated by a vector field $Z\in\mathfrak{X}(Q)$ vanishing on the
 boundary of $\phi(U_0)$, then we conclude that $\phi$ is an extremal
 if and only if
\[\ds\int_{U_0} \left((\mathcal{L}_{Z^c} L)
\circ \phi^{(1)}\right)(x) d^kx=0\,, \] for all
$Z$ satisfying the above conditions, where $Z^c$ is the complete
lift of $Z$ to $T^1_kQ$. Putting $Z=Z^i\derpar{}{q^i}$,   from (\ref{clift}),   we know that   the local expression of the complete lift $Z^c$ is
$$
Z^c=Z^i\frac{\partial}{\partial q^i}+v^i_\alpha \frac{\partial Z^j}
{\partial q^i}\frac{\partial}{\partial v^j_\alpha} \,.
$$
Then   
integrating by parts we deduce that
$\phi(x)=(\phi^i(x))$ is an extremal of
$\mathcal{J}$ if and only if
\[\ds\int_{U_0} \left[\displaystyle \sum_{\alpha=1}^k\ds\frac{\partial}{\partial x^\alpha}\Big\vert_{x}
\left(\frac{\displaystyle\partial L}{\displaystyle
\partial v^i_\alpha}\Big\vert_{\phi^{(1)}(x)} \right)- \frac{\displaystyle \partial
L}{\displaystyle
\partial q^i}\Big\vert_{\phi^{(1)}(x)} \right]Z^i d^kx=0\,, \]for all values of $Z^i$.
Thus, $\phi$ will be an extremal of $\mathcal{J}$ if and only if

\begin{equation}\label{ELe}
\displaystyle \sum_{\alpha=1}^k\ds\frac{\partial}{\partial
x^\alpha}\Big\vert_{x} \left(\frac{\displaystyle\partial
L}{\displaystyle
\partial v^i_\alpha}\Big\vert_{\phi^{(1)}(x)} \right)= \frac{\displaystyle \partial
L}{\displaystyle
\partial q^i}\Big\vert_{\phi^{(1)}(x)} \;.
\end{equation} The equations (\ref{ELe}) are called the  {\it
Euler-Lagrange field equations} for the Lagrangian $L\in\mathcal{C}^\infty(T^1_kQ)$.

\subsubsection{The intrinsic form of the Euler-Lagrange field equations}

We shall give an intrinsic  form of the  Euler-Lagrange field equations (\ref{ELe}), 
using the canonical $1$-form $\lambda\in T^1_k(\tkqh)$ introduced in Section \ref{chilambda}.

In order to do this we consider a map  $\psi:U_0\subset  \rk \to \tkqh $     locally given by  $\psi(x)=(\psi^i(x),\psi_i^\alpha(x))$,  then $\phi=\pi^k_Q \circ \psi : U_0\subset  \rk \to Q$ is a  map locally given by  $\phi(x)=(\psi^i(x))$.

Let us recall that the local coordinates on $T^1_k(\tkqh)$ are  $(q^i,p^\alpha_i,(v_\alpha)^i,(v_\alpha)^\beta_i)
$ then  the  first prolongation   $$\psi^{(1)} :U_0\subset  \rk \to  T^1_k(\tkqh)  $$    of $\psi$  
 has the local expression 
   \begin{equation}\label{psi0-1}
    \psi^{(1)}(x)=(\psi^i(x),\psi_i^\beta(x),  \derpar{\psi^i}{x^\alpha}\Big\vert_x
   , \derpar{\psi_i^\beta }{x^\alpha}\Big\vert_x  ) .
 \end{equation}

     \begin{prop}\label{lambda_EL}
     
     Let $L\in \mathcal{C}^\infty(T^1_kQ)$ be a Lagrangian function. 
   If $\psi\colon U_0\subset \mathbb{R}^k\to (T^1_k)^*Q$ satisfies \begin{equation}\label{first1} \left( \lambda -   (T^1_k\pi_Q^k)^*(dL)\right)  \left(\psi^{(1)}(x)\right)   = 0
   \end{equation} then $\phi=\pi^k_Q \circ \psi$ is a solution to the Euler-Lagrange field equations (\ref{ELe}).
  
   \end{prop}
     \proof

First we compute $ (T^1_k\pi_Q^k)^*(dL)$.  Since  $T^1_k\pi^k_Q \colon   T^1_k (\tkqh)
    \lto    T^1_kQ$  is locally given by
 \begin{equation}\label{t1kp1}
   T^1_k\pi^k_Q \left(q^i,p^\alpha_i,(v_\alpha)^i,(v_\alpha)^\beta_i  \right)
=  (q^i,  (v_\alpha)^ i ),
\end{equation}
from    (\ref{psi0-1}) and  (\ref{t1kp1})   we have that
 \begin{equation}\label{t1kp}
    T^1_k\pi^k_Q \circ \psi^{(1)}= \phi^{(1)},
     \end{equation}    where $\phi = \pi^k_Q\circ \psi :U_0\subset  \rk \to Q$ is  locally given by $\phi(x)=(\psi^i(x))\, .$

        On the other hand
        \begin{equation}\label{hhelp7}
dL
=\derpar{L}{q^i}  dq^i    +    \derpar{L}{v^i_\alpha}  dv^i_\alpha,     \end{equation}
and from  (\ref{t1kp}) and    (\ref{hhelp7}),  we obtain
    \begin{equation}\label{hhelp8}
    \left[(T^1_k\pi_Q^k)^*(dL)\right]\left(\psi^{(1)}(x)\right) = \derpar{L}{q^i}\Big\vert_{\left(\phi^{(1)}(x)\right)} \,  \,  dq^i ( \psi^{(1)}(x)  )
   +    \derpar{L}{v^i_\alpha}\Big\vert_{\left(\phi^{(1)}(x)\right)}  \,  \,   d(v_\alpha)^i ( \psi^{(1)}(x)  ).
    \end{equation}

     From (\ref{eq_local_lambda})
we have
     \begin{equation}\label{lambchi}
\lambda\left(\psi^{(1)}(x)\right)=\left( \ds\sum_{\alpha=1}^k \derpar{\psi^\alpha_i}{x^\alpha}\Big\vert_{ x } \right)  \, dq^i(\psi^{(1)}(x))
   +   \psi_i^\alpha(x)\,  d(v_\alpha)^i (\psi^{(1)}(x))\,.
   \end{equation}
   
Now if $\psi$ satisfies (\ref{first1}), from      (\ref{hhelp8}) and (\ref{lambchi}), we deduce that $\psi$ is solution to the equations
$$
 \ds\sum_{\alpha=1}^k \derpar{\psi^\alpha_i}{x^\alpha}\Big\vert_{ x } =\derpar{L}{q^i}\Big\vert_{\left(\phi^{(1)}(x)\right)} , \quad  \psi_i^\alpha(x)= \derpar{L}{v^i_\alpha}\Big\vert_{\left(\phi^{(1)}(x)\right)},
$$
 and therefore $\phi $   is a solution to the  Euler-Lagrange field equations (\ref{ELe}).

The equations (\ref{first1}) will be called the \textit{intrinsic form of the Euler-Lagrange field equations}.

     \begin{remark} The converse of the Proposition \ref{lambda_EL} is not true in general. If  $\psi\colon U_0\subset \mathbb{R}^k\to (T^1_k)^*Q$ is a map such that $\phi=\pi^k_Q\circ \psi$ is a solution of the Euler-Lagrange field equations (\ref{ELe}), then $\psi$ satisfies the equation (\ref{first1}) if and only if $\psi=FL\circ \phi^{(1)}$.
     
     Grabowska in \cite{Grabowska} characterizes the  Euler-Lagrange equations   (\ref{ELe}) defining certain subset of $T^1_k(\tkqh))$, which is obtained 
 starting from  a map $\alpha:T^1_k(\tkqh)\to T^*(\tkq)$, generalization of the Tulczyjew isomorphism $\alpha_Q: T(T^*Q)\to T^*(TQ)$  defined in
\cite{T3}. For a  non autonomous  $L(x^\alpha,q^i,v^i_\alpha)$
     a similar proposition is given in \cite{RRSV}.
\end{remark}
\subsection{Implicit Euler-Lagrange field equations}

In order to describe the implicit Euler-Lagrange field equations it is necessary to consider the $k$-Pontryagin bundle $\mathcal{M}=T^1_kQ\oplus_Q(T^1_k)^*Q$. We now introduce the Hamilton-Pontryagin variational principle and then, we consider the intrinsic version of the implicit Euler-Lagrange field equation, in this case using the canonical form $\chi$.

In the last part of this subsection we consider one particular example: Navier's equations.

\subsubsection{The Hamilton-Pontryagin principle}

Using the canonical forms $\Theta^1_\mathcal{M}, \ldots, \Theta^k_\mathcal{M}$ and the generalized energy Lagrangian function, defined in Section \ref{subsec:Pontryagin}, we establish the Hamilton-Pontryagin principle for $k$-symplectic classical field theories. This principle is similar to the expression for the multisymplectic case \cite{VYL-2012}.

Consider  a Lagrangian function $L\in \mathcal{C}^\infty(T^1_kQ)$     with associated generalized energy function $E$.  We now define the Hamilton-Pontryagin  action functional 
\begin{equation}\label{def:Pontryagin_action}
 \begin{array}{rccl}
 \mathcal{S}\colon & \mathcal{C}^\infty_C(\mathbb{R}^k,\mathcal{M}) & \to &\mathbb{R}\\\noalign{\medskip}
  & \Psi & \mapsto &\mathcal{S}(\Psi)=\displaystyle\int_{\mathbb{R}^k} \Big(\Psi^*(\Theta^\alpha_{\mathcal{M}})\wedge d^{k-1}x^\alpha - \Psi^*(E)d^kx\Big)\, ,
 \end{array}
\end{equation}
where $\mathcal{C}^\infty_C(\mathbb{R}^k,\mathcal{M})$ is the set of maps $\Psi=(\phi^i,\phi^i_\alpha, \psi^\alpha_i)\colon K\subset U_0\subset \mathbb{R}^k\to \mathcal{M}=T^1_kQ\oplus_Q(T^1_k)^*Q$, with compact support $K$, defined on an open set $U_0$.

\begin{remark}\label{remark:psi}
Let us observe that $\Psi$ can be written as $\Psi=(pr_1\circ\psi,pr_2\circ\psi)$ where $pr_1\circ\Psi\colon K\subset U_0\subset \mathbb{R}^k\to T^1_kQ$ and $pr_2\circ\Psi\colon K\subset U_0\subset \mathbb{R}^k\to (T^1_k)^*Q$
\end{remark}

Employing local coordinates $(q^i,v^i_\alpha,p^\alpha_i)$ on $\mathcal{M}=T^1_kQ\oplus_Q (T^1_k)^*Q$, the action functional is denoted by
\[
\mathcal{S}({\rm v}_q,\nu_q) = \displaystyle\int_{\mathbb{R}^k} \left[
p^\alpha_i\left(\displaystyle\frac{\partial q^i}{\partial x^\alpha} - v^i_\alpha\right)+ L(q^i,v^i_\alpha)
\right]d^kx\,.
\]

Let us observe that in the case $k=1$ we obtain the Hamilton-Pontryagin action functional introduced in \cite{YM-2006b}.

We have defined the Hamilton-Pontryagin action functional in terms of the family of $1$-forms $\Theta^1_\mathcal{M}, \ldots, \Theta^k_\mathcal{M}$. By using the definitions (\ref{eq:local_can_form_M}) and (\ref{eq:energy}) we can rewrite the Hamilton-Pontryagin action functional as:
\[
\begin{array}{lcl}
\mathcal{S}(\Psi)&=&\mathcal{S}(pr_1\circ \Psi, pr_2\circ \Psi)
\\\noalign{\medskip} &=&\displaystyle\int_{\mathbb{R}^k}\left(<<(\tau^k_Q\circ pr_1\circ\Psi)^{(1)}, pr_2\circ\Psi>> - << pr_1\circ\Psi,pr_2\circ\Psi>> + L(pr_1\circ\Psi) \right)d^kx\,,
\end{array}
\]
where $\Psi=(pr_1\circ \Psi, pr_2\circ \Psi)\colon K\subset U_0\subset\mathbb{R}^k\to \mathcal{M}= T^1_kQ\oplus_Q(T^1_k)^*Q$.

A map $\Psi\in \mathcal{C}^\infty_C(\mathbb{R}^k,\mathcal{M})$ is an extremal of the above action if
\[
\displaystyle\frac{d}{ds}\Big\vert_{s=0}\mathcal{S}(\tau^1_k\tau_s\circ\Psi)=0\,,
\]for each flow $\tau_s\colon Q\to Q$ such that $\tau_s(q)=q$ for every $q\in pr^\mathcal{M}_Q(\Psi(\partial K))$. Since such a flow $\tau_s$ is generated by a vector field $Z\in\mathfrak{X}(Q)$   vanishising at all points of $pr^\mathcal{M}_Q( \Psi(\partial K))$, we can prove that $\Psi$ is an extremal of $\mathcal{S}$ if and only if
 \[
  \displaystyle\int_{\mathbb{R}^k} [\Psi^*(\mathcal{L}_{Z^1}E)]d^kx =0\,.
 \]
for all $Z$ satisfying the above conditions, where $Z^1$ is the complete lift of $Z$ to $\mathcal{M}$.

We now suppose that $\Psi=(\phi^i,\phi^i_\alpha, \psi^\alpha_i)$ satisfies  $pr_1\circ \Psi=(pr^\mathcal{M}_Q\circ \Psi)^{(1)}$ that is, in a local coordinate system,
\[\phi^i_\alpha(x)=\displaystyle\frac{\partial \phi^i}{\partial x^\alpha}\Big\vert_{x},\; 1\leq i\leq n, \makebox{ and } 1\leq \alpha\leq k .\]

Consider now the canonical coordinate systems such that $Z=Z^i\nicefrac{\partial}{\partial q^i}$; taking into account the local expression (\ref{local:prontryagin_lift}) for the complete lift $Z^1$, the local expression of the generalized energy (\ref{local:energy}) and that $\Psi(x)=(\phi^i(x), \phi^i_\alpha(x)=\nicefrac{\partial \phi^i}{\partial x^\alpha}, \psi^\alpha_i(x))$, we deduce that $\psi$ is an extremal of $\mathcal{S}$ if and only if
\[
\begin{array}{lcl}
 \displaystyle\int_{\mathbb{R}^k}Z^i(x)\left ( \displaystyle\sum_{\alpha=1}^k\displaystyle\frac{\partial\psi^\alpha_i}{\partial x^\alpha}\Big\vert_{x} - \displaystyle\frac{\partial L}{\partial q^i}\Big\vert_{pr_1(\Psi(x))}\right)d^kx &=& 0 \, ,
 \\\noalign{\medskip}
 \displaystyle\int_{\mathbb{R}^k} \left[\displaystyle\sum_{\alpha=1}^k\displaystyle\frac{\partial (Z^i\circ pr^{\mathcal{M}}_Q\circ \psi)}{\partial x^\alpha}\Big\vert_{x} \left(\psi^\alpha_i(x)-\displaystyle\frac{\partial L}{\partial v^i_\alpha}\Big\vert_{pr_1(\psi(x))}\right)\right]d^kx&=& 0  \, ,
\end{array}
\]for all  $Z^i$ and $\nicefrac{\partial Z^i}{\partial q^j}$, with $Z$ vanishing at all points of $pr^\mathcal{M}_Q( \Psi(\partial K))$. Thus, $\Psi$ will be an extremal of $\mathcal{S}$ if and only if
\begin{equation}\label{eq_Pontryagin_Lagrange_eq}
\phi^i_\alpha(x)=\displaystyle\frac{\partial \phi^i}{\partial x^\alpha}\Big\vert_{x},\quad \displaystyle\sum_{\alpha=1}^k\displaystyle\frac{\partial\psi^\alpha_i}{\partial x^\alpha}\Big\vert_{x} - \displaystyle\frac{\partial L}{\partial q^i}\Big\vert_{pr_1(\Psi(x))} =0 \quad , \quad \psi^\alpha_i(x)-\displaystyle\frac{\partial L}{\partial v^i_\alpha}\Big\vert_{pr_1(\Psi(x))}=0 \, .
\end{equation}

\noindent These equations are called \textit{the implicit Euler-Lagrange field equations.}
\begin{remark}\begin{enumerate}
  \item The first group of the equations (\ref{eq_Pontryagin_Lagrange_eq}) shows that,
 $pr_1(\Psi(x))=  \phi^{(1)}(x) $
where
$\phi :\rk \to \phi(x)=\tau^k_Q(\Psi(x)) =(\phi^i(x))\in Q$.

\item The last group of the equations (\ref{eq_Pontryagin_Lagrange_eq}) implies that, in the conditions of the above proposition,
$FL(pr_1\circ \Psi) = pr_2\circ \Psi\,.$

\item From (\ref{eq_Pontryagin_Lagrange_eq}) we deduce that $\phi$ is solution to the Euler-Lagrange equations
(\ref{ELe})
\end{enumerate}
\end{remark}

\subsubsection{The intrinsic form of the implicit Euler-Lagrange field equations}

We shall give an intrinsic  characterization of the  Euler-Lagrange equations (\ref{eq_Pontryagin_Lagrange_eq}), 
using the intrinsic $1$-form $\chi$ introduced in Section \ref{chilambda} .

Let  $\Psi:U_0\subset  \rk \to  \mathcal{M}$    be a map with local expression $ \Psi(x)=\left(\phi^i(x),\phi^i_\alpha(x),\psi^\alpha_i(x)\right)$ and    let  $\Psi^{(1)} :U_0\subset  \rk \to  T^1_k\mathcal{M}$ be  its first prolongation,   which is locally given by
   \begin{equation}\label{Psi0-1}
     \Psi^{(1)}(x)=\left(\phi^i(x),\phi^i_\alpha(x),\psi^\alpha_i(x), \derpar{\phi^i}{x^\beta} ,
      \derpar{\phi^i_\alpha}{x^\beta},    \derpar{\psi^\alpha_i}{x^\beta}    \right)  \, ,
 \end{equation}
 see (\ref{localfi11}) .

 \begin{prop}\label{tul_implicit}
Let $\Psi:\rk \to \mathcal{M}$ be a map, and $\Psi^{(1)}:\rk \to T^1_k\mathcal{M}$ its first prolongation. Then
 $\psi $ satisfies
\begin{equation}\label{justin0}
\left[ (T^1_kpr_2)^* \chi -(\tau^k_{\mathcal{M}})^*dE
 \right]
\left(\Psi^{(1)}(x)\right)=0
\end{equation}
if and only if $\Psi$ is solution to the  implicit Euler-Lagrange equations  (\ref{eq_Pontryagin_Lagrange_eq}).
\end{prop}
 
\proof

Since
$T^1_kpr_2 : T^1_k\mathcal{M}    \to   T^1_k(\tkqh)$  is locally given by

 $$T^1_kpr_2 \left(q^i,v^i_\alpha,p^\alpha_i,(v_\alpha)^i,(v_\alpha)^i_\beta,(v_\alpha)^\beta_i\right) =  \left(q^i,p^\alpha_i,(v_\alpha)^i,(v_\alpha)^\beta_i\right)
 $$
we obtain  from   the local expressions        (\ref{eq:local_chi})   and    (\ref{Psi0-1})   that

\begin{equation}\label{0or1ksim}
\begin{array}{ccl}
(T^1_kpr_2)^* \chi\left(\Psi^{(1)}(x)\right)
 =   \derpar{\phi^i}{x^\alpha}\Big\vert_{x}\,dp^\alpha_i(\Psi^{(1)}(x))\, -
   \derpar{\psi^\alpha_i}{x^\alpha}\Big\vert_{x}\, dq^i(\Psi^{(1)}(x))\, .
    \end{array}
\end{equation}

On the other hand, since the canonical projection
$\tau^k_{\mathcal{M}}\colon T^1_k\mathcal{M}  \to   \mathcal{M}$ is locally given by
\begin{equation}\label{or1ksim}
   \tau^k_{\mathcal{M}}\left(q^i,v^i_\alpha, p^\alpha_i,(v_\alpha)^i,(v_\alpha)^i_\beta,(v_\alpha)^\beta_i\right)  =(q^i,v^i_\alpha,p^\alpha_i) ,
\end{equation}  we deduce from (\ref{local:energy}) 
that
\begin{equation}\label{or1ksim2}
\begin{array}{rl}
(\tau^k_{\mathcal{M}  } )^*(dE)(\Psi^{(1)}(x))=  &-  \derpar{L}{q^i}\Big\vert_{pr_1(\Psi(x))}d q^i (\Psi^{(1)}(x))\\\noalign{\medskip} &+(
 \psi^\alpha_i (x) -\derpar{L}{v^i_\alpha}\Big\vert_{pr_1(\Psi(x))} )\, d v^i_\alpha (\Psi^{(1)}(x))
 +      \phi^i_\alpha(x)     \, d p^\alpha_i (\Psi^{(1)}(x))\,.
 \end{array}
 \end{equation}

From      (\ref{0or1ksim}) and     (\ref{or1ksim2})  we deduce that  $\Psi$ is solution to (\ref{justin0})
 if and only if   $    \Psi$
is solution to the implicit Euler-Lagrange equations (\ref{eq_Pontryagin_Lagrange_eq}).
\qed
 The equations (\ref{justin0}) are called \textit{the intrinsic characterization of the implicit Euler-Lagrange field equations.}
\begin{remark}
The above proposition is a generalization of   Proposition 3.3 in \cite{YM-2006b}.
\end{remark}

%
%
%

\subsubsection{Example: Navier's equations}

\paragraph {A. Navier's equations  \cite{olver}}\

The deformations of an elastic body  $\Omega \subset \r^n$
 are described by the
displacement field $\phi\colon  \Omega \to \r^n$.
  Each material point   $ x\in \Omega$ in the undeformed body will
move to a new position $y=x+\phi(x)$ in the derformed body.

The one-dimensional case governs bars, beams and rods, two-dimensional bodies include
thin plates and shells, while $n = 3$ for fully three-dimensional solid bodies.

%
%

The simplest case is that of a homogeneous
and isotropic planar body $\Omega \subset \r^2$,
being  the deformation function (field)
 $$\phi  \colon (x^1,x^2) \in \Omega  \to   \phi(x^1,x^2) =  ( \phi^1(x^1,x^2),  \phi^2(x^1,x^2))\in Q\equiv \r^2  \, . $$

 Navier's equations  for $\phi$ are
\begin{equation} \label{eq _navier}\begin{array}{lcl}
(\lambda+2\mu) \partial_{11}\phi^1   +
(\lambda+\mu)\partial_{12}\phi^2  +
\mu \partial_{22}\phi^1  &=&0\,,\\\noalign{\medskip}
\mu \partial_{11}\phi^2 +(\lambda+\mu)\partial_{12}\phi^1
+ (\lambda+2\mu) \partial_{22}\phi^2   &=&0\,,
\end{array} \end{equation}
where the parameters $\lambda,\, \mu$ are known as the {\it Lam\'e moduli} of the material, and govern its intrinsic
elastic properties, see \cite{Antman}  and \cite{Gurtin}  for
details and physical derivations. In the above equations we use the notation $\partial_{\alpha\beta}\phi^i=\partial^2\phi^i/\partial x^\alpha\partial x^\beta$.

  Navier's equations can seen as a particular case of the  Euler-Lagrange equations 
 (\ref{ELe}) for the    Lagrangian  $L:T^1_2\r^2=T\r^2\oplus_{\r^2 }T\r^2\to \r$ given by
    \begin{equation}\label{lagrangian_navier}
L  (q^1,q^2, v^1_1,v^2_1,v^1_2, v^2_2)=
(\ds\frac{1}{2}\lambda + \mu) \left[(v_1^1)^2+ (v_2^2)^2\right]+
\ds\frac{1}{2}\mu \left[(v^1_2)^2 + (v^2_1)^2\right] +
(\lambda+\mu)v^1_1v^2_2 \, ,\end{equation}
with $(q^1,q^2, v^1_1,v^2_1,v^1_2, v^2_2)$  coordinates on $T^1_2\r^2$.

 \paragraph{B.    Implicit equations}\

Now we shall write Navier's equations (\ref{eq _navier}) using the intrinsic equation (\ref{justin0})  with the  Lagrangian
  (\ref{lagrangian_navier}).


 Employing local coordinates
 $(q^1,q^2,v^1_1,  v^2_1 ,   v^1_2,  v^2_2, p^1_1, p^1_2, p^2_1, p^2_2)$
  on $\mathcal{M}=T^1_k\r^2\oplus_{\r^2} (T^1_k)^*\r^2$, each map  $\Psi\colon  U_0\subset \mathbb{R}^2\to  \mathcal{M} $  will be locally  written
as follows
$$\Psi=
(\phi^1,\phi^2,\phi^1_1,  \phi^2_1 ,   \phi^1_2,  \phi^2_2, \psi^1_1, \psi^1_2, \psi^2_1,\psi^2_2) \, .
$$

From (\ref{eq:local_chi}), we know that    $\chi =  - ((v_1)^1_i+  (v_2)^2_i  ) \,\, dq^i  +   (v_1)^i\, \,  dp^1_i + (v_2)^i\, \,  dp^2_i
 $,
 and we have in this case
 \begin{equation}\label{0or1ksim0}
 \begin{array}{rl}
(T^1_kpr_2)^* \chi\left(\Psi^{(1)}(x)\right)
 =  &  - \left( \derpar{\psi^1_i}{x^1}\Big\vert_{x} +\derpar{\psi^2_i}{x^2}\Big\vert_{x}\right)\, dq^i(\Psi^{(1)}(x))\\\noalign{\medskip}
 & + \derpar{\phi^i}{x^1}\Big\vert_{x}\,dp^1_i(\Psi^{(1)}(x))\, +   \derpar{\phi^i}{x^2}\Big\vert_{x}\,dp^2_i(\Psi^{(1)}(x))\,.\end{array}
   \end{equation}
   
Since in this example  
 $dE=
(p^1_i- \ds\frac{\partial L}{\partial v^i_1})dv^i_1 + (p^2_i- \ds\frac{\partial L}{\partial v^i_2})dv^i_2 +
v^i_1 dp^1_i+
v^i_2 dp^2_i\, ,
$
we have that
\begin{equation}\label{deejemplo}
\begin{array}{rl}
 (\tau^k_{ \mathcal{M}} )^*(dE)(\Psi^{(1)}(x))= & 
\left( \psi^\alpha_i (x) -\derpar{L}{v^i_\alpha}   \Big\vert_{pr_1(\Psi(x))} \right)\, d v^i_\alpha (\Psi^{(1)}(x))\\\noalign{\medskip}
 &+      \phi^i_1(x)     \, d p^1_i (\Psi^{(1)}(x))  + \phi^i_2(x)     \, d p^2_i (\Psi^{(1)}(x))\, .
 \end{array}
 \end{equation}
 
Now, from (\ref{lagrangian_navier}), we obtain    $$ \derpar{L}{v^1_1}= ( \lambda+2\mu)v^1_1 + (\lambda+\mu)v^2_2,
\, ,\quad
 \derpar{L}{v^1_2} = \mu v^1_2\, ,
   \quad    \derpar{L}{v^2_1}= \mu v^2_1 \, , \quad          \derpar{L}{v^2_2}=(\lambda+2\mu) v^2_2 + (\lambda+\mu)v^1_1\, ,$$
and, from (\ref{0or1ksim0}) and   (\ref{deejemplo}), we deduce that  the implicit field equations in this example are
\begin{equation}\label{pont_Navier}\begin{array}{c}
 \phi^1_\alpha= \derpar{\phi^1}{x^1}\, ,\quad   \phi^2_\alpha= \derpar{\phi^2}{x^\alpha}\, ,\quad 0=  \displaystyle\frac{\partial \psi^1_1}{\partial x^1}+\displaystyle\frac{\partial \psi^2_1}{\partial x^2}, \quad
 0
  = \displaystyle\frac{\partial \psi^1_2}{\partial x^1}+\displaystyle\frac{\partial \psi^2_2}{\partial x^2}
\\ \noalign{\medskip}
\psi^1_1=(\lambda+2\mu)\phi^1_1+(\lambda+\mu)\phi^2_2,  \quad
\psi^1_2= \mu \phi^2_1\, ,\\ \noalign{\medskip} \psi^2_1=\mu\phi^1_2,\quad
 \psi^2_2=(\lambda+2\mu)\phi^2_2+(\lambda+\mu)\phi^1_1 \, .\end{array}\end{equation}
Finally from (\ref{pont_Navier})  we obtain

 $$\begin{array}{l}
 0=  \displaystyle\frac{\partial \psi^1_1}{\partial x^\alpha}+\displaystyle\frac{\partial \psi^2_1}{\partial x^2}=
 (\lambda+2\mu)   \derpar{^2\phi^1}{(x^1)^2}   + (\lambda+\mu)\ds\frac{\partial \phi^2}{\partial x^1 \partial x^2}
 +\mu  \derpar{^2\phi^1}{(x^2)^2}\, ,
 \\ \noalign{\medskip}
 0
  = \displaystyle\frac{\partial \psi^1_2}{\partial x^1}+\displaystyle\frac{\partial \psi^2_2}{\partial x^2}
  = \mu   \derpar{^2\phi^2}{(x^1)^2} + (\lambda+2\mu)  \derpar{^2\phi^2}{(x^2)^2}
  +  (\lambda+\mu)  \ds\frac{\partial \phi^1}{\partial x^1 \partial x^2}\, ,
\end{array}
$$
which give   the Navier's equations   (\ref{eq _navier}) .

\subsection{Non-holonomic implicit Euler-Lagrange field equations}\label{subsec:constraints}

We now consider the implicit Euler-Lagrange field equations with non-holonomic constraints. The results of the above section can be considered a particular case of this section.
\subsubsection{The Lagrange-D'Alembert-Pontryagin principle}
In this section we study the case in which a regular constraint distribution is given. To do this, we introduce an extended Lagrange-D'Alembert principle called the  Lagrange-D'Alembert-Pontryagin principle.

In first place we describe this situation and we introduce the necessary geometric elements, a complete description of this elements can be found in \cite{LMSV-2008}.

We consider a field theory built on the following geometric objects:
\begin{itemize}
\item A Lagrangian function $L\colon T^1_kQ\to \mathbb{R}$.
\item A constraint submanifold $\mathcal{N}\hookrightarrow T^1_kQ$, which can be locally represented by equations of the form $\Phi_A(q^i,v^i_\alpha)=0$ for $A=1, \ldots, m$. This submanifold represents some external constraints imposed on the system. For the sake of clarify we will confine ourselves to the case that $\mathcal{N}$ projects onto the whole of $Q$ and the restriction $\tau^k_Q\vert_\mathcal{N}\colon \mathcal{N}\to Q$ is a fibre bundle.
\item A bundle $\mathcal{F}$ of constraints forms, defined along $\mathcal{N}$, where $\mathcal{F}$ is generated by the $m$ independent semi-basic $\mathbb{R}^k$-valued one-forms $\eta_1,\ldots, \eta_m$ that locally read \begin{equation}\label{exp:constraints}
\eta_A=(\eta^1_A, \ldots, \eta^k_A)=((\eta^1_{A})_idq^i,\ldots, (\eta^k_{A})_idq^i)\,,
\end{equation}
for some smooth function $(\eta^\alpha_A)_i$ on $\mathcal{N}\subset T^1_kQ$. The independence of the forms $\eta_A$ clearly implies that the $m\times kn$-matrix whose elements are the functions $(\eta^\alpha_A)_i$, has constant maximal rank $m$.
\end{itemize}

We now have to specify the field equations. Proceeding as in the case of unconstrained $k$-symplectic field theories, we consider the following definition:
\begin{definition}\label{def:LDP_problem}
Let $L\in \mathcal{C}^\infty(T^1_kQ)$ be a Lagrangian function,
$\mathcal{N}\hookrightarrow T^1_kQ$ a constraint submanifold, $\mathcal{F}$ the bundle of constraints forms defined along $\mathcal{N}$. If   $\Psi\in \mathcal{C}^\infty_C(\mathbb{R}^k,\mathcal{M})$ is a map, with compact support $K$, defined onto a open set $U_0$ such that $pr_1\circ \Psi=(pr^\mathcal{M}_Q\circ\Psi)^{(1)}$ and $(pr_1\circ\Psi)(U_0)\subset\mathcal{N}$, then $\Psi$ is a solution of the \textit{Lagrange-D'Alembert-Pontryagin problem} if
for each $Z\in \mathfrak{X}(Q)$ which vanishes on  $(pr^\mathcal{M}_Q\circ\Psi)(\partial K)$ and such that $\iota_{Z^C}\eta=0$ for all $\eta$ of the bundle of constraints forms $\mathcal{F}$, we have
\[
\displaystyle\int_{\mathbb{R}^k}[\Psi^*(\mathcal{L}_{Z^1}E)]d^kx=0\,,\]
where $Z^C$ and $Z^1$ are the complete lift of $Z$ to $T^1_kQ$ and $\mathcal{M}=T^1_kQ\oplus_Q(T^1_k)^*Q$ respectively.
\end{definition}

Let $L\in \mathcal{C}^\infty(T^1_kQ)$ be a Lagrangian function,
$\mathcal{N}\hookrightarrow T^1_kQ$ a constraint submanifold and $\mathcal{F}$ the bundle of constraints forms defined along $\mathcal{N}$ and  $\psi\in \mathcal{C}^\infty_C(\mathbb{R}^k,\mathcal{M})$ a map with compact support $K$, defined onto a open set $U_0$ and such that $pr_1\circ \psi=(pr^\mathcal{M}_Q\circ\psi)^{(1)}$ and $(pr_1\circ\psi)(U_0)\subset\mathcal{N}$. 
It is easy to prove that $\psi$ is a solution of the  Lagrange-D'Alembert-Pontryagin problem if and only if
\[
\begin{array}{lcl}
 \displaystyle\int_{\mathbb{R}^k}Z^i(x)\left ( \displaystyle\sum_{\alpha=1}^k\displaystyle\frac{\partial\psi^\alpha_i}{\partial x^\alpha}\Big\vert_{x} - \displaystyle\frac{\partial L}{\partial q^i}\Big\vert_{pr_1(\psi(x))}\right)d^kx &=& 0 \, ,
 \\\noalign{\medskip}
 \displaystyle\int_{\mathbb{R}^k} \left[\displaystyle\sum_{\alpha=1}^k\displaystyle\frac{\partial Z^i}{\partial q^j}\Big\vert_{x} \displaystyle\frac{\partial \phi^j}{\partial x^\alpha}\Big\vert_{x} \left(\psi^\alpha_i(x)-\displaystyle\frac{\partial L}{\partial v^i_\alpha}\Big\vert_{pr_1(\psi(x))}\right)\right]d^kx&=& 0\, ,
\end{array}
\]for all $Z\in \mathfrak{X}(Q)$ satisfying $\iota_{Z^C}\eta=0$ for all $\eta\in \mathcal{F}$  and, thus for any values $Z^i$ and $\nicefrac{\partial Z^i}{\partial q^j}$ such that
\[(\eta^\alpha_A)_iZ^i=0,\; 1\leq A\leq m,\; 1\leq \alpha \leq k.\]

Thus, $\psi$ is solution of the  Lagrange-D'Alembert-Pontryagin problem if and only if $\psi=(\phi^i,\phi^i_\alpha, \psi^\alpha_i)$ is a solution to the following systems of partial differential equations:
\begin{equation}\label{eq:constraint_implecit_EL_eq}
\begin{array}{l}
\displaystyle\frac{\partial \phi^i}{\partial x^\alpha}\Big\vert_{x} = \phi^i_\alpha(x)\, , \quad \psi^\alpha_i(x)=\displaystyle\frac{\partial L}{\partial v^i_\alpha}\Big\vert_{pr_1(\psi(x))}\, ,\quad \quad \Phi_A((pr^\mathcal{M}_Q\circ\psi)^{(1)})(x))=0 \, , \\\noalign{\medskip}
\displaystyle\sum_{\alpha=1}^k\displaystyle\frac{\partial \psi^\alpha_i}{\partial x^\alpha}\Big\vert_{x} - \displaystyle\frac{\partial L}{\partial q^i}\Big\vert_{pr_1(\psi(x))}=\lambda^A_\alpha (\eta^\alpha_{A})_i(pr_1(\psi(x))) \,.
\end{array}
\end{equation}
with $1\leq i\leq n,\, 1\leq \alpha\leq k, \, 1\leq A\leq m$. Here the functions $\lambda^A_\alpha$ play the role of the Lagrange multipliers.

The equations (\ref{eq:constraint_implecit_EL_eq}) are called \textit{the nonholonomic implicit Euler-Lagrange field equations}.

\subsubsection{The intrinsic form}
 
 We shall give an intrinsic  characterization of the non-holonomic implicit Euler-Lagrange equations (\ref{eq:constraint_implecit_EL_eq}), 
using the intrinsic $1$-form $\chi$ introduced in Section \ref{chilambda}  and the constraint  forms $\eta_1,\ldots, \eta_k$.

 \begin{prop}\label{tul_implicit-pontri}
Let $\Psi:\rk \to \tm$ be a map, such that  $\Phi_A((pr^\mathcal{M}_Q\circ\Psi)^{(1)})(x))=0$. Then
 $\Psi $ satisfies
\begin{equation}\label{justin00}
\left[  (T^1_kpr_2)^* \chi  -
 (\tau^k_{\mathcal{M}})^* \left[ dE  -(pr_1)^*( \lambda^A_\alpha \eta^\alpha_{A})  \right]     \right]\left(\Psi^{(1)}(x)\right)=0
\end{equation}
if and only if $\Psi$ is solution to    (\ref{eq:constraint_implecit_EL_eq})
\end{prop}
\proof It is similar to the proof of Propositions \ref{lambda_EL} and \ref{tul_implicit}.

\qed

\subsubsection{Examples}\label{examples}

In this part of the paper we consider two examples of the nonholonomic implicit Euler-Lagrange equations. The first one in the non-holonomic Cosserat rod. This example is an example of second  order field theory, but it is possible to consider a description of this example with an associated first-order Lagrangian function. The second one is the particular case of linear constraints induced by a family of distributions on Q.

\paragraph{A. The nonholonomic Cosserat rod}\

The nonholonomic Cosserat rod \cite{ST-2007,Van-2012,VM-2008} is an example of a nonholonomic field theory which describes the motion of a rod  which is constrained to roll without sliding on a horizontal surface.

\begin{figure}[h!]
\centering
\includegraphics[scale=0.45]{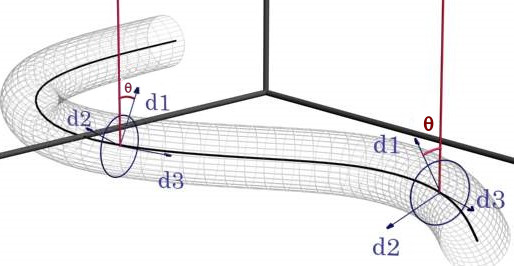}
\caption{Cosserat rod}
\end{figure}

A Cosserat rod can be thought of as a long and thin deformable body. We assume that its length is significantly larger that its radius. A Cosserat rod can be visualized as a curve $s\to r(s)\in\mathbb{R}^3$, called \textit{centerline}, to which is attached a frame $\{\mathbf{d}_1(s), \mathbf{d}_2(s), \mathbf{d}_3(s)\}$ called \textit{director frame}. In this orthonormal basis, the vector $\mathbf{d}_3(s)$ is constrained to be parallel to $r'(s)$.

We consider an inextensible Cosserat rod of lenght $l$. Is we denote the centerline at $t$ as $s\to \mathbf{r}(t,s)$, inextensibility allows us to assume that the parameter $s$ is
the arc length.

The nonholonomic second-order model of the Cosserat rod is described into the multisymplectic framework. The complete description can be found in \cite{Van-2012}.  In \cite{LMSV-2008} we modified this model by a lowering process to obtain a first-order Lagrangian function. In this model we consider the Lagrangian function
\begin{equation}\label{Cosserat_lagrangian}
L=\displaystyle\frac{\rho}{2}(\dot{x}^2 + \dot{y}^2) + \displaystyle\frac{\alpha}{2}\dot{\theta}^2 -\displaystyle\frac{1}{2}(\beta(\theta ')^2 + K((z')^2+ (v')^2)) + \lambda(z-x') + \mu (v-y')\,.
\end{equation}

Here $\rho, \alpha, \beta$ and $K$ are real parameters, $(x(t,s),y(t,s))$ are the coordinates of the centerline, $\theta(t,s)$ is the torsion angle, $\dot{x}=\partial x/\partial t$, $x'=\partial x/\partial s$ (analogous for $y$ and $\theta$) and $\lambda $ and $\mu$ are Lagrange multipliers associated to the constraint $z=x'$ and $v=y'$. The constraints are given by
\begin{equation}\label{constraints}
\dot{x}+R\dot{\theta}y'=0\quad  \makebox{and} \quad \dot{y}-R\dot{\theta}x'=0\,,
\end{equation} where  $R$ is another real parameter.

The Lagrangian (\ref{Cosserat_lagrangian}) can be thought as a mapping defined on $T^1_2Q$ where $Q= \mathbb{R}^2\times \mathbb{S}^1\times \mathbb{R}^4\equiv \mathbb{R}^7$.  If we rewrite this Lagrangian with the notation of Section \ref{subsec:k-tangent} we obtain that the first order Lagrangian $L\colon T^1_2\mathbb{R}^7\to \mathbb{R}$  is given by
\[
L=\displaystyle\frac{\rho}{2}((v^1_1)^2 + (v^2_1)^2) + \displaystyle\frac{\alpha}{2}(v^3_1)^2 -\displaystyle\frac{1}{2}(\beta(v^3_2)^2 + K((v^4_2)^2+ (v^5_2)^2)) + q^6(q^4-v^1_2) + q^7 (q^5-v^2_2)\, ,
\]
subject to constraints
\[
v^1_1 + Rv^3_1v^2_2=0\quad \makebox{and}\quad v^2_1-Rv^3_1v^1_2=0\,.
\]

Let us observe that this Lagrangian is a singular Lagrangian. In this case, the constraint submanifold is the set
\[
\mathcal{N}=\{{\rm v}_q=({v_1}_q, {v_2}_q)\in T^1_2\mathbb{R}^7 / \Phi_1({\rm v}_q)=v^1_1+Rv^3_1v^2_1=0 \quad \makebox{and}\quad \Phi_2({\rm v}_q)=v^2_1-Rv^3_1v^1_2=0\}\,,
\]
where ${v_\alpha}_q=\sum_{i=1}^7v_\alpha^i\partial/\partial q^i$. Therefore, the bundle of constraints forms $\mathcal{F}$ is generated by the $2$ $\mathbb{R}^2$-valued $1$-forms
\[
\eta_1=(\eta_1^1, \eta_1^2)=(dq^1 + Rv^2_2dq^3, 0) \quad \makebox{ and }\quad \eta_2=(\eta_2^1,\eta_2^2)=(dq^2-Rv^1_2dq^3, 0)\,.
\]

A solution of the Lagrange-D'Alembert-Pontryagin problem in this case is a map
\[
\begin{array}{rccl}
\Psi\colon & \mathbb{R}^2 & \to & \mathcal{M}=T^1_2\mathbb{R}^7\oplus_{\mathbb{R}^7}(T^1_2)^*\mathbb{R}^7\\\noalign{\medskip}
 & x=(t,s) & \mapsto & (\phi^i(x) , \phi^i_1(x), \phi^i_2(x), \psi^1_i(x),\psi^2_i(x)) \, ,
\end{array}
\] with $1\leq i\leq 7$, which satisfies the following system of partial differential equations:
\begin{equation}\label{cosserat_eq}
\begin{array}{c}
\rho \displaystyle\frac{\partial^2 \phi^1}{\partial t\partial t} - \displaystyle\frac{\partial \phi^6}{\partial s}= \nu_1, \quad \rho \displaystyle\frac{\partial^2 \phi^2}{\partial t\partial t} - \displaystyle\frac{\partial \phi^7}{\partial s}= \nu_2, \quad \alpha\displaystyle\frac{\partial^2 \phi^3}{\partial t\partial t}- \beta \displaystyle\frac{\partial^2 \phi^3}{\partial s\partial s} = R\left(\nu_1\displaystyle\frac{\partial \phi^2}{\partial s}-\nu_2 \displaystyle\frac{\partial \phi^1}{\partial s}\right)\\\noalign{\medskip}
\phi^4=\displaystyle\frac{\partial \phi^1}{\partial s},\quad \phi^5=\displaystyle\frac{\partial \phi^2}{\partial s},\quad
\phi^6= -K\displaystyle\frac{\partial^2\phi^4}{\partial s\partial s}, \quad  \phi^7= -K\displaystyle\frac{\partial^2\phi^5}{\partial s\partial s}
\\\noalign{\medskip}
\phi^i_1= \displaystyle\frac{\partial \phi^i}{\partial t},\quad \phi^i_2= \displaystyle\frac{\partial \phi^i}{\partial s} \qquad (1\leq i \leq 7)
\\\noalign{\medskip}
\psi^1_1=\rho\displaystyle\frac{\partial \phi^1}{\partial t}, \quad
\psi^1_2=\rho\displaystyle\frac{\partial \phi^2}{\partial t}, \quad
\psi^1_3=\alpha\displaystyle\frac{\partial \phi^3}{\partial t}, \quad
\psi^1_4=\psi^1_5=\psi^1_6=\psi^1_7=0
\\\noalign{\medskip}
\psi^2_1=-\phi^6,\quad \psi^2_2=-\phi^7,\quad \psi^2_3=-\beta\displaystyle\frac{\partial \phi^3}{\partial s},\quad
\psi^2_4=-k\displaystyle\frac{\partial \phi^4}{\partial s},\quad
\psi^2_5=-k\displaystyle\frac{\partial \phi^5}{\partial s},\quad
\psi^2_6=\psi^2_7=0
\\\noalign{\medskip}
\phi^1_1+R\phi^3_1\phi^2_2=0,\quad \phi^2_1+R\phi^3_1\phi^1_2=0\,,
\end{array}
\end{equation}
where $\nu_1$ and $\nu_2$ are Lagrange multipliers associated with the nonholonomic constraints, $x=(t,s)$ are the time and space coordinates,   and  the field $\phi\colon \mathbb {R}^2\to\mathbb{R}^7$ is given by  the coordinates of the centerline $(\phi^1(t,s), \phi^2(t,s))$ and by the torsion angle $\phi^3(t,s)$. As one can see in Eq. (\ref{cosserat_eq}) the components $\phi^i,\, i\geq 4$ are determined by $(\phi^1, \phi^2,\phi^3)$.

The equations (\ref{cosserat_eq}) can be written as in (\ref{justin00}), that is, in a intrinsic form. In fact, in this particular case we have
\begin{equation}\label{eq:cos_intrinsic_1}
\begin{array}{l}
[(T^1_2pr_2)^*\chi](\Psi^{(1)}(x)) =
\\\noalign{\medskip}  \qquad
\displaystyle\frac{\partial \phi^1}{\partial t}\Big\vert_{x}dp^1_1 + \displaystyle\frac{\partial \phi^1}{\partial s}\Big\vert_{x}dp^2_1 - \Big(\displaystyle\frac{\partial \psi^1_1}{\partial t}\Big\vert_{x}+\displaystyle\frac{\partial \psi^2_1}{\partial 2}\Big\vert_{x} \Big)dq^1
\\\noalign{\medskip}\quad  + 
\displaystyle\frac{\partial \phi^2}{\partial t}\Big\vert_{x}dp^1_2 + \displaystyle\frac{\partial \phi^2}{\partial s}\Big\vert_{x}dp^2_2 - \Big(\displaystyle\frac{\partial \psi^1_2}{\partial t}\Big\vert_{x}+\displaystyle\frac{\partial \psi^2_2}{\partial 2}\Big\vert_{x} \Big)dq^2
\\\noalign{\medskip}\quad   + 
\displaystyle\frac{\partial \phi^3}{\partial t}\Big\vert_{x}dp^1_3 + \displaystyle\frac{\partial \phi^3}{\partial s}\Big\vert_{x}dp^2_3 - \Big(\displaystyle\frac{\partial \psi^1_3}{\partial t}\Big\vert_{x}+\displaystyle\frac{\partial \psi^2_3}{\partial 2}\Big\vert_{x} \Big)dq^3
\\\noalign{\medskip}\quad   + 
\displaystyle\frac{\partial \phi^4}{\partial t}\Big\vert_{x}dp^1_4 + \displaystyle\frac{\partial \phi^4}{\partial s}\Big\vert_{x}dp^2_4 - \Big(\displaystyle\frac{\partial \psi^1_4}{\partial t}\Big\vert_{x}+\displaystyle\frac{\partial \psi^2_4}{\partial 2}\Big\vert_{x} \Big)dq^4
\\\noalign{\medskip}\quad   + 
\displaystyle\frac{\partial \phi^5}{\partial t}\Big\vert_{x}dp^1_5 + \displaystyle\frac{\partial \phi^5}{\partial s}\Big\vert_{x}dp^2_5 - \Big(\displaystyle\frac{\partial \psi^1_5}{\partial t}\Big\vert_{x}+\displaystyle\frac{\partial \psi^2_5}{\partial 2}\Big\vert_{x} \Big)dq^5
\\\noalign{\medskip}\quad   + 
\displaystyle\frac{\partial \phi^6}{\partial t}\Big\vert_{x}dp^1_6 + \displaystyle\frac{\partial \phi^6}{\partial s}\Big\vert_{x}dp^2_6 - \Big(\displaystyle\frac{\partial \psi^1_6}{\partial t}\Big\vert_{x}+\displaystyle\frac{\partial \psi^2_6}{\partial 2}\Big\vert_{x} \Big)dq^6
\\\noalign{\medskip}\quad   + 
\displaystyle\frac{\partial \phi^7}{\partial t}\Big\vert_{x}dp^1_7 + \displaystyle\frac{\partial \phi^7}{\partial s}\Big\vert_{x}dp^2_7 - \Big(\displaystyle\frac{\partial \psi^1_7}{\partial t}\Big\vert_{x}+\displaystyle\frac{\partial \psi^2_7}{\partial 2}\Big\vert_{x} \Big)dq^7 \, ,
\end{array}
\end{equation}
and
\begin{equation}\label{eq:cos_intrinsic_2}
\begin{array}{l}
[(\tau^2_{\mathcal{M}})^*[dE-(pr_1)^*(\lambda^A_\alpha\eta^\alpha_A)]](\Psi^{(1)}(x))=  
\\\noalign{\medskip} 
\quad -\lambda^1_1dq^1 - \lambda^2_1dq^2- R(\lambda^1_1\phi^2_2(x)-\lambda^2_1\psi^1_2(x))dq^3 - \phi^6(x)dq^4 
\\\noalign{\medskip}
\quad -\phi^7(x)dq^5 - (\phi^4(x)-\phi^1_2(x))dq^6 - (\phi^5(x)-\phi^2_2(x))dq^7
\\\noalign{\medskip}
\quad + 
(\psi^1_1(x)-\rho\phi^1_1(x))dv^1_1 + (\psi^2_1(x)+\phi^6(x)) dv^1_2
\\\noalign{\medskip}
\quad + 
(\psi^1_2(x)-\rho\phi^2_1(x))dv^2_1 + (\psi^2_2(x)+\phi^7(x)) dv^2_2
\\\noalign{\medskip}
\quad + 
(\psi^1_3(x)-\alpha\phi^3_1(x))dv^3_1 + (\psi^2_3(x)+\beta\phi^3_2(x)) dv^3_2
\\\noalign{\medskip}
\quad + \psi^1_4(x)dv^4_1 + (\psi^2_4(x) + K\phi^4_2(x))dv^4_2
\\\noalign{\medskip}
\quad + \psi^1_5(x)dv^5_1 + (\psi^2_5(x) + K\phi^5_2(x))dv^5_2
\\\noalign{\medskip}
\quad + \psi^1_6(x)dv^6_1 + \psi^2_6(x) dv^6_2+ \psi^1_7(x)dv^7_1 + \psi^2_7(x) dv^7_2
\\\noalign{\medskip}
\quad
 + \phi^1_1(x)dp^1_1  + \phi^1_2(x)dp^2_1 + \phi^2_1(x)dp^1_2  + \phi^2_2(x)dp^2_2 
 \\\noalign{\medskip}
\quad
 + \phi^3_1(x)dp^1_3  + \phi^3_2(x)dp^2_3  + \phi^4_1(x)dp^1_4  + \phi^4_2(x)dp^2_4
  \\\noalign{\medskip}
\quad
 + \phi^5_1(x)dp^1_5  + \phi^5_2(x)dp^2_5  + \phi^6_1(x)dp^1_6 + \phi^6_2(x)dp^2_6 
 \\\noalign{\medskip}
\quad
 + \phi^7_1(x)dp^1_7  + \phi^7_2(x)dp^7_5\,.
\end{array}
\end{equation}

Now, from (\ref{justin00}), (\ref{eq:cos_intrinsic_1}) and (\ref{eq:cos_intrinsic_2}) we obtain the equations of the non-holonomic Cosserat rod (\ref{cosserat_eq}).

\paragraph{B. Linear constraints induced by distributions on $Q$}\

Let $D_1,\ldots, D_k$ be $k$ distributions on $Q$. We now consider the constraint submanifold $\mathcal{N}$ defined by $\mathcal{N}=D_1\oplus_Q\cdots\oplus_QD_k$ of $T^1_kQ$.

We now assume, for each $\alpha $ with $1\leq \alpha\leq k$, that the annihilator $D_\alpha^0$ of each distribution $D_\alpha$ is spanned by the $1$-forms on $Q$ locally given by
\begin{equation}\label{eq:form_distribution}
\bar\psi^\alpha_{l_\alpha}=(\bar\psi^\alpha_{l_\alpha})_idq^i,\; l_\alpha=1,\dots, m_\alpha,\end{equation}  where $(\bar\psi^\alpha_{l_\alpha})_i$ is a family of functions defined on $Q$. In this situation, $D_\alpha$ is the set of solutions to the $m_\alpha$ equations \[\Psi^\alpha_{l_\alpha}(v_q)\colon =\bar\psi^\alpha_{l_\alpha}(q) (v_q) =(\bar\psi^\alpha_{l_\alpha})_i(q)v^i=0\,.\]

Thus, $D_\alpha$ is defined by the vanishing of $m_\alpha$ independent functions $\Psi^\alpha_{l_\alpha}\in\mathcal{C}^\infty(TQ)$, that is,
\begin{equation}\label{eq:linear_contraint}
(D_\alpha)(q)\colon =\{v_q\in T_qQ / \Psi^\alpha_{l_\alpha}(v_q)=(\bar\psi^\alpha_{l_\alpha})_iv^i=0,\, 1\leq l_\alpha\leq m_\alpha\}\subset T_qQ\,.
\end{equation}
Thus, the constraint submanifold $\mathcal{N}$ is given by the vanishing of $m=m_1+\cdots+m_k$ independent functions $\Phi^\alpha_{l_\alpha}$ where
\[
\Phi^\alpha_{l_\alpha} (v_{1_q},\ldots, v_{k_q})\colon=[(\tau^{k,\alpha}_Q)^*\Psi^\alpha_{l_\alpha}](v_{1_q},\ldots, v_{k_q})= (\bar\psi^\alpha_{l_\alpha})_i(q)v^i_\alpha\,.
\]

The bundle of constraints forms $\mathcal{F}$
is generated by the $m$  $\mathbb{R}^k$-valued  basic $1$-forms on $T^1_kQ$, with $m=m_1+m_2+\ldots +m_k$

$$\begin{array}{cclcl}
\eta^1_{l_1}&=& \left(   (\tau^k_Q)^*\bar\psi^1_{l_1} ,0,\ldots, 0 \right)&=&\left(   (\bar\psi^1_{l_1})_idq^i,0 ,\ldots, 0\right)
\\ \noalign{\medskip}
\eta^2_{l_2}&=&  \left(  0,(\tau^k_Q)^*\bar\psi^2_{l_2} ,0,\ldots, 0\right)&=&\left(0,  (\bar\psi^2_{l_2})_idq^i,0 ,\ldots, 0\right)
\\ \noalign{\medskip}  \ldots &  & \ldots & & \ldots
\\ \noalign{\medskip}
\eta^k_{l_k}&=&  \left(  0,\ldots, 0,(\tau^k_Q)^*\bar\psi^k_{l_k}  \right)&=&\left(0, \ldots, 0 ,(\bar\psi^k_{l_k})_idq^i\right)
\end{array}
$$

 with $i=1,\ldots, n,\, \alpha = 1,\ldots, k,$ and $ l_\alpha= 1,\ldots, m_\alpha\,.$

Thus 
the implicit nonholonomic Euler-Lagrange field equations (\ref{eq:constraint_implecit_EL_eq}),   are in this case
 
\begin{equation}\label{eq:linear_cont}\begin{array}{l}
\displaystyle\frac{\partial \phi^i}{\partial x^\alpha}\Big\vert_{x} = \phi^i_\alpha(x),\, \quad \psi^\alpha_i(x)=\displaystyle\frac{\partial L}{\partial v^i_\alpha}\Big\vert_{\phi^{(1)}(x)},\quad \quad \Phi^\alpha_{l_\alpha}(\phi^{(1)}(x))=0\\\noalign{\medskip}
\displaystyle\sum_{\alpha=1}^k\displaystyle\frac{\partial \psi^\alpha_i}{\partial x^\alpha}\Big\vert_{x}- \displaystyle\frac{\partial L}{\partial q^i}\Big\vert_{\phi^{(1)}(x)}=\lambda^{l_\alpha}_\alpha (\psi^\alpha_{l_\alpha})_i((\phi(x))) \,,
\end{array}
 \end{equation}
 \begin{remark}
In this particular case, a map $\Psi\colon \mathbb{R}^k\to \mathcal{M}$, with local expression $\Psi(x)=(\phi^i(x),\phi^i_\alpha(x), $ $\psi^\alpha_i(x))$ is a solution of (\ref{eq:linear_cont}), then $\phi(x)=(\phi^i(x))$ is a solution of the nonholonomic field equation associated to linear constraints induced by distributions on $Q$ (see \cite{LMSV-2008}, page 818, for more details about these equations.).
\end{remark}

Also, it is possible to write the intrinsic form of the equations (\ref{eq:linear_cont}) using the characterization given in Proposition \ref{tul_implicit-pontri}.
\section{Hamilton-de Donder-Weyl equations}\label{sec:HDWeq}
In this section we consider the Hamilton-de Donder-Weyl equations. The idea is to give the intrinsic form of these equations in two cases, without constraints and with non-holonomic constraints.

\subsection{Classical Hamilton-de Donder-Weyl equations}
In a similar way that in the Lagrangian approach, we can describe the Hamilton-de Donder-Weyl equations from a variational principle.

Along this subsection we consider an arbitrary Hamiltonian function $H\in\mathcal{C}^\infty((T^1_k)^*Q)$.

We define the functional
\[
\mathcal{H}(\psi)=\int_{\mathbb{R}^k}\left(\displaystyle\sum_{\alpha=1}^k [\psi^*[(\pi^{k,\alpha}_Q)^*\theta]]\wedge d^{k-1}x_\alpha - \psi^*Hd^kx 
\right) \,,
\]
where $\theta\in \Lambda^1(T^*Q)$ is the canonical Liouville form, $\pi^{k,\alpha}_Q\colon (T^1_k)^*Q\to T^*Q$   the projection defined in (\ref{map:taukalpha}) and $\psi\colon U_0\subset\mathbb{R}^k\to (T^1_k)^*Q$ is a map with compact support $K$.

A map $\psi$ is an extremal of the above action if
\[
\displaystyle\frac{d}{ds}\mathcal{H}(\sigma_s\circ \psi)\Big\vert_{s=0}=0 \,,
\]
for every flow $\sigma_s$ on $(T^1_k)^*Q$ such that $\sigma_s(\nu_q)= \nu_q$ for all $\nu_q\in \psi(\partial K).$

In \cite{LSV-2015} we proved that $\psi$ is and extremal of $\mathcal{H}$ if and only if $\psi$ is a solution of the Hamilton-De Donder-Weyl equations, that is, if $\psi$ is locally given by $\psi(x)=(\psi^i(x), \psi^\alpha_i(x))$, then the functions $\psi^i$ and $\psi^\alpha_i$ satisfy the system of partial differential equations
\begin{equation}\label{eq:HDW}
\displaystyle\frac{\partial H}{\partial q^i}\Big\vert_{\psi(x)} = -\, \displaystyle\sum_{\alpha=1}^k\displaystyle\frac{\partial \psi^\alpha_i}{\partial x^\alpha}\Big\vert_{x},\qquad
\displaystyle\frac{\partial H}{\partial p^\alpha_i}\Big\vert_{\psi(x)} = \displaystyle\frac{\partial \psi^i}{\partial x^\alpha}\Big\vert_{x}\,.
\end{equation}

We now give a characterization of the Hamilton-de Donder-Weyl equations (\ref{eq:HDW}) using the canonical $1$-form $\chi$ introduced in Section \ref{sec:Tulczyjew}.

\begin{prop}\label{prop10}
A map   $\psi:U_0\subset \mathbb{R}^k \to (T^1_k)^*Q$ is solution to the Hamilton-de Donder-Weyl equations (\ref{eq:HDW}) if and only if
\begin{equation}\label{justin02}
\left[\chi  - (\tau^k_{\tkqh}) ^*(dH) \right] \left(\psi^{(1)}(x)\right)=0 \, , \quad x \in \r^k \, .
\end{equation}
\end{prop}
\begin{proof}

If $\psi\colon U_0\subset\mathbb{R}^k\to  (T^1_k)^*Q$ is  a map locally given by $\psi(x)=(\psi^i(x), \psi^\alpha_i(x))$, then the first prolongation $\psi^{(1)}\colon U_0\subset\mathbb{R}^k\to T^1_k((T^1_k)^*Q)$ has the local expression (\ref{psi0-1}).

First we compute $\left[\chi  - (\tau^k_{\tkqh}) ^*(dH) \right] \left(\psi^{(1)}(x)\right)$. We know that
\begin{equation}\label{eq:dh}
dH=\displaystyle\frac{\partial H}{\partial q^i}dq^i + \displaystyle\frac{\partial H}{\partial p^\alpha_i}dp^\alpha_i\,.
\end{equation}

From (\ref{eq:local_chi}), (\ref{psi0-1}) and  (\ref{eq:dh}) we obtain
\[
\begin{array}{l}
\left[\chi  - (\tau^k_{\tkqh}) ^*(dH) \right] \left(\psi^{(1)}(x)\right) = \\\noalign{\medskip}
\left( -\, \displaystyle\sum_{\alpha=1}^k\displaystyle\frac{\partial \psi^\alpha_i}{\partial x^\alpha}\Big\vert_{x} -\displaystyle\frac{\partial H}{\partial q^i}\Big\vert_{\psi(x)} \right)dq^i\left(\psi^{(1)}(x)\right)  + \left( \displaystyle\frac{\partial \psi^i}{\partial x^\alpha}\Big\vert_{x} -\displaystyle\frac{\partial H}{\partial p^\alpha_i}\Big\vert_{\psi(x)} \right)dp^\alpha_i\left(\psi^{(1)}(x)\right)\,.
\end{array}
\]

Therefore, $\psi$ satisfies (\ref{justin02}) if and only if 
\[\displaystyle\frac{\partial H}{\partial q^i}\Big\vert_{\psi(x)} = -\, \displaystyle\sum_{\alpha=1}^k\displaystyle\frac{\partial \psi^\alpha_i}{\partial x^\alpha}\Big\vert_{x},\qquad
\displaystyle\frac{\partial H}{\partial p^\alpha_i}\Big\vert_{\psi(x)} = \displaystyle\frac{\partial \psi^i}{\partial x^\alpha}\Big\vert_{x}\,,\] that is, if and only if $\psi$ is a solution of the Hamilton-de Donder-Weyl equations.
\end{proof}

The expression (\ref{justin02}) is called \textit{the intrinsic form of the Hamilton-de Donder Weyl equations in $(T^1_k)^*Q$.}

\begin{remark}
This Proposition is a generalization of the Proposition 3.11 of \cite{YM-2006b}.
\end{remark}

Let us observe that in the previous result the Hamiltonian function is an arbitrary function defined on the cotangent bundle of $k^1$-covelocities $(T^1_k)^*Q$. An interesting case is when the Hamiltonian function is given by a hyperregular Lagrangian $L$.

Given a hyperregular Lagrangian $L\in \mathcal{C}^\infty (T^1_kQ)$, the Legendre transformation $FL\colon T^1_kQ \to (T^1_k)^*Q$ is a global diffeomorphism, then a hyperregular Hamiltonian $H\in \mathcal{C}^\infty ((T^1_k)^*Q)$ can be defined by
\begin{equation}\label{hamiltonian}
H= E_L\circ (FL)^{-1} \, , 
\end{equation}
where $E_L$ is the energy function, with local expression
\[
E_L(q^i, v^i_\alpha)=\displaystyle\frac{\partial L}{\partial v^i_\alpha} v^i_\alpha - L(q^i, v^i_\alpha)\,.
\]

In this case we can say that the above proposition is the dual of   Proposition \ref{lambda_EL}. 
\subsection{Intrinsic form of the non-holonomic  Hamilton-de Donder-Weyl equations}

In this subsection we consider the Hamilton-de Donder-Weyl equations when a constraint submanifold $\mathcal{N}\subset T^1_kQ$ is given. We again consider the geometric objects described in Subsection \ref{subsec:constraints}. The only different is that in this case we consider a hyperregular Lagrangian function $L$. Then we can consider the Hamiltonian $H$ defined by the expression  (\ref{hamiltonian}).

If the constraint submanifold $\mathcal{N}$ is locally represented by a family of $m$ equations of the form $\Phi_A(q^i,v^i_\alpha)=0, \, 1\leq A \leq m$, then the constraint function on $(T^1_k)^* Q$ becomes $\Psi_A = \Phi_A\circ FL^{-1}\colon (T^1_k)^*Q\to \mathbb{R}\,.$
\begin{prop}
A map   $\psi:U_0\subset \mathbb{R}^k \to (T^1_k)^*Q$, such that $\Psi_A(\psi(x))=0, \, 1\leq A\leq m$, satisfies
\begin{equation}\label{intHDWeq}
\left[\chi  - (\tau^k_{\tkqh})^*\Big (dH-\lambda^A_\alpha (FL^{-1})^*\eta^\alpha_A\Big) \right] \left(\psi^{(1)}(x)\right)=0 \, , \quad x \in \r^k \, ,
\end{equation}
if and only $\psi$ is a solution of the system of partial differential equations given by
\begin{equation}\label{eq:nonholonomic_Ham}
\begin{array}{rcl}
\displaystyle\frac{\partial H}{\partial q^i}\Big\vert_{\psi(x)} &=& -\, \displaystyle\sum_{\alpha=1}^k\displaystyle\frac{\partial \psi^\alpha_i}{\partial x^\alpha}\Big\vert_{x} + \lambda^A_\alpha(\eta^\alpha_A)_i (FL^{-1}(\psi(x)) \, ,  \\\noalign{\medskip}
\displaystyle\frac{\partial H}{\partial p^\alpha_i}\Big\vert_{\psi(x)} &=& \displaystyle\frac{\partial \psi^i}{\partial x^\alpha}\Big\vert_{x}\,.
\end{array}
\end{equation}

\end{prop}

\begin{proof}
It's similar to the proof of   Proposition \ref{prop10}.
\end{proof}

The equations (\ref{eq:nonholonomic_Ham}), joint with the conditions $\Psi_A(\psi(x))=0, \, 1\leq A\leq m$, are the \textit{non-holonomic Hamilton-de Donder-Weyl equations on $(T^1_k)^*Q$} defined in \cite{LMSV-2008}. Therefore, the equations (\ref{intHDWeq}) are called \textit{the intrinsic non-holonomic Hamilton-de Donder-Weyl equations.}

\begin{remark}
The above proposition is the Hamiltonian counterpart of the Proposition \ref{tul_implicit-pontri}. 
\end{remark}

\section{Conclusions}\label{sec:conclusions}
This work presents the variational principles and the  intrinsic versions of several equations in field theories, in particular, for the Classical Euler-Lagrange field equations, the implicit Euler-Lagrange field equations and the non-holonomic implicit Euler-Lagrange field equations. The advantages of the variational and intrinsic versions of these equations is that the Lagrangians functions are not necessary regular Lagrangians. 

In particular,   equation (\ref{justin00}) should be highlighted, since all the cases described in this work could be considered as particular case of the one described by this equation.

We also present two examples of the results of this work: Navier's equations and non-holonomic Cosserat rod.

The key to being able to write these equations in an intrinsic way has been to define, using Tulczyjew's derivations, two canonical forms $\lambda$ and $\chi$.
  
Finally we present the Hamiltonian counterpart of these results, in particular when the Hamiltonian function is defined from a hyper-regular Lagrangian function.
  
In Mechanics, the implicit Euler-Lagrange field equations can be obtained using Dirac structures \cite{YM-2006a,YM-2006b}. 
In fact, some of our results are a generalization of the results of these works. For this reason, we think that as a future work, it will be interesting to analyse if Dirac's structures  can help to obtain new descriptions of the Euler-Lagrange field equations.

\subsection*{Acknowledgements}
Modesto Salgado acknowledges the financial support from the 
Spanish Ministerio de Econom\'ia y Competitividad
project MTM2014--54855—P and
the Ministerio de Ciencia, Innovaci\'on y Universidades project
PGC2018-098265-B-C33.

Silvia Vilariño acknowledges partial financial support from 
the  Spanish Ministerio de Ciencia, Innovaci\'on y Universidades project PGC2018-098265-B-C31; from the Spanish Ministerio de Econom\'{\i}a y Competitividad project MTM2015--64166--C2-1-P;
and from the Aragon Government grant E38$\_$17R.


\end{document}